%% file: thesis.tex
\documentclass[letterpaper,12pt,titlepage,openright,twoside,final]{book}
 
\newcommand{\href}[1]{#1} 

\usepackage{ifthen}
\newboolean{PrintVersion}
\setboolean{PrintVersion}{true} 

\usepackage{amsmath,amssymb,amstext} 
\usepackage{graphicx} 

\usepackage[letterpaper=true,pagebackref=false]{hyperref} 
\hypersetup{
    plainpages=false,       
    pdfpagelabels=true,     
    bookmarks=true,         
    unicode=false,          
    pdftoolbar=true,        
    pdfmenubar=true,        
    pdffitwindow=false,     
    pdfstartview={FitH},    
    pdftitle={Master Thesis},    
        pdfauthor={Abel Molina},    
    pdfsubject={Computer Science},  
    pdfnewwindow=true,      
    colorlinks=true,        
    linkcolor=blue,         
    citecolor=green,        
    filecolor=black,      
    urlcolor=pink          
}

\let\origdoublepage\cleardoublepage
\newcommand{\clearemptydoublepage}{%
  \clearpage{\pagestyle{empty}\origdoublepage}}
\let\cleardoublepage\clearemptydoublepage

\input{00-includes}

\begin{document}

\input{00-header}



\input{01-intro}

\input{02-background}

\input{03-formalism}

\input{04-average}

\input{05-notindependent}

\input{06-upperbound}

\input{07-errorreduction}

\input{08-conclusion}

\cleardoublepage
\phantomsection
\addcontentsline{toc}{chapter}{Bibliography}
\bibliographystyle{alpha}
\bibliography{thesis}

\end{document}

%% file: 00-includes.tex
\usepackage{amsmath}

\usepackage{amssymb} 
\usepackage{amsthm} 
\usepackage{bbm} 

\usepackage{eepic} 

\DeclareMathOperator{\fid}{F}

\newenvironment{mylist}[1]{\begin{list}{}{
	\setlength{\leftmargin}{#1}
	\setlength{\rightmargin}{0mm}
	\setlength{\labelsep}{2mm}
	\setlength{\labelwidth}{8mm}
	\setlength{\itemsep}{0mm}}}
	{\end{list}}

\newtheorem{theorem}{Theorem}
\newtheorem{lemma}[theorem]{Lemma}

\theoremstyle{definition}
\newtheorem{definition}[theorem]{Definition}

\newcommand{\ket}[1]{|#1\rangle}

\newcommand{\op}[1]{\operatorname{#1}}

\newcommand{\reg}[1]{\mathsf{#1}}

\newcommand{\lin}[1]{\setft{L}\left(#1\right)}

\newcommand{\tr}{\operatorname{Tr}}

\newcommand{\tinyspace}{\mspace{1mu}}

\newcommand{\norm}[1]{\left\lVert\tinyspace #1 \tinyspace\right\rVert}

\newcommand{\ip}[2]{\left\langle #1 , #2\right\rangle}
\newcommand{\setft}[1]{\mathrm{#1}}

\newcommand{\density}[1]{\setft{D}\left(#1\right)}

\newcommand{\herm}[1]{\setft{Herm}\left(#1\right)}
\newcommand{\pos}[1]{\setft{Pos}\left(#1\right)}
\newcommand{\channel}[1]{\setft{C}\left(#1\right)}
\newcommand{\pd}[1]{\setft{Pd}\left(#1\right)}

\def\X{\mathcal{X}}
\def\Y{\mathcal{Y}}
\def\Z{\mathcal{Z}}
\def\W{\mathcal{W}}

\newcommand{\floor}[1]{\left\lfloor #1 \right\rfloor}

%% file: 00-header.tex
\pagestyle{empty}
\pagenumbering{roman}

\begin{titlepage}
        \begin{center}
        \vspace*{1.0cm}

        \Huge
        {\bf Parallel Repetition of Prover-Verifier Quantum Interactions}

        \vspace*{1.0cm}

        \normalsize
        by \\

        \vspace*{1.0cm}

        \Large
        Abel Molina \\

        \vspace*{3.0cm}

        \normalsize
        A thesis \\
        presented to the University of Waterloo \\ 
        in fulfillment of the \\
        thesis requirement for the degree of \\
        Master of Mathematics \\
        in \\
        Computer Science and Quantum Information\\

        \vspace*{2.0cm}

        Waterloo, Ontario, Canada, 2011 \\

        \end{center}
\end{titlepage}

\pagestyle{plain}
\setcounter{page}{2}

%
%
%

\cleardoublepage

\begin{center}\textbf{Abstract}\end{center}

In this thesis, we answer several questions about the behaviour of prover-verifier interactions under parallel repetition when quantum information is allowed, and the verifier acts independently in them. 

We first consider the case in which a value is associated with each of the possible outcomes of an interaction. We prove that it is not possible for the prover to improve on the optimum average value per repetition by repeating the protocol multiple times in parallel.

We look then at games in which the outcomes are classified into two types, winning outcomes and losing outcomes. We ask what is the optimal probability for the prover of winning at least $k$ times out of $n$ parallel repetitions, given that the optimal probability of winning when only one repetition is considered is $p$. A reasonable conjecture for the answer would be $\sum_{m \geq k}  {n \choose m} p^m (1-p)^{n-m}$, as that is the answer when it is optimal for the prover to act independently. This is known to be the correct answer when $k=n$, and also in the classical case. It is also correct in some generalizations of the classical case that we will discuss later. We will show how this cannot be extended to all cases, presenting an example of an interaction with $k=1,n=2$ in which $p\approx 0.85$, but it is possible to always win at least once. We will then give some upper bounds on the optimal probability for the prover of winning $k$ times out of $n$ parallel repetitions. These bounds are expressed as a function of $p$.

Finally, we will connect our results to the study of error reduction for quantum interactive proofs using parallel repetition.

\newpage

\cleardoublepage


\begin{center}\textbf{Preface}\end{center}

The work presented in this thesis has been performed while I was enrolled as a Masters student at the University of Waterloo. This work has been performed under the supervision of John Watrous.

I would like to thank my friends for all I have learned from them along these years, and for keeping me relatively sane.

I would like to thank my family for supporting me, specially when I decided to transfer abroad during the middle of my undergraduate degree.

I would like to thank John Watrous for all his support, specially for giving  me the opportunity to learn from him when I approached him as an undergraduate student that did not know much.

I would like to thank Ray Laflamme and Michele Mosca for creating the Institute for Quantum Computing, and make it the leading institution in research that it is today. I would like as well to thank Mike and Ophelia Lazaridis for their continued support for the institution.

Finally, I  would like to thank the readers of the thesis, Richard Cleve and Ashwin Nayak, for their time and their useful comments. I would also like to thank Gus Gutoski for his help with figures 3.1 and 3.2 in this thesis,  David Rhee for his help at analyzing the objective value of the last dual solution in Chapter 6, and Thomas Vidick for suggesting the first of the four open questions mentioned in the conclusion.

\newpage

\tableofcontents
\newpage




\pagenumbering{arabic}

%% file: 01-intro.tex

\chapter{Introduction} \label{ch:intro}

We will give now an abstract description of the kind of interaction that we consider, without giving details about the corresponding underlying theories or mathematical structures. In our prover-verifier interaction, one individual (the prover) subjects another individual (the verifier) to a test. Following the standard convention for two-party interactions in quantum information, we will call them Alice and Bob, respectively. They could also be named Arthur and Merlin, following the convention in computational complexity for prover-verifier interactions. Of course, Alice and Bob might correspond 
to devices instead of individuals in a real life instance of this kind of interaction.

The interaction corresponding to our tests is of the following form:
\begin{mylist}{\parindent}
\item[1.]
Alice prepares a \emph{question} and sends it to Bob.

\item[2.]
Bob responds by sending an \emph{answer} to Alice.

\item[3.]
The previous steps are repeated an arbitrary number of times. At any point, Alice and Bob can use whatever memory they have of the interaction to determine what question or answer to send.

\item[4.]

Based on the last answer from Bob, as well as whatever memory Alice has of Bob's previous answers and her own questions, Alice assigns an outcome to the test.

\end{mylist}

An interaction is specified by the process by which Alice operates. We assume that Bob has access to a complete description of this process. In a classical setting, the messages exchange between Alice and Bob are purely classical, that is, they can be modelled as a sequence of bits. The process by which Alice operates can then be modelled as  a given probabilistic process, where the questions are selected from some probability distribution conditioned on previous parts of the interaction, and the final decision might also involve the use of randomness. In the quantum case, Alice's questions might take the form of quantum information, and so might the answers that she expects from Bob. The process by which Alice operates is then at each step a given map from quantum states to quantum states. This process transforms Alice's memory and the last answer he received from Bob to Alice's next question and the next state of her memory. Note that this implies that Alice's questions can be entangled with Alice's memory.

The behaviour of Bob is not part of the description of the test. Indeed, the questions we explore in this thesis are mostly concerned with looking at what behaviour is desirable for Bob in different cases. Typically Bob is allowed to perform an arbitrary probabilistic process in the classical case, and an arbitrary quantum process in the quantum case. In the same way as the process for Alice, this process can be conditioned on previous parts of the interaction, and in the quantum case Bob can entangle his answers with his memory.

Note also that there is no loss of generality involved in assuming that the protocol begins with a message from Alice. This is because we can simulate a similar protocol in which the first message is sent by Bob with a protocol of the kind described here in which the first message is sent by Alice, and it is an empty message.

The formalism necessary to study these interactions in a rigorous way is presented in Chapters \ref{ch:back} and \ref{ch:sdp}. Chapter \ref{ch:back} presents some useful linear algebra, optimization and quantum information facts and terms. Chapter \ref{ch:sdp} shows how these can be applied to obtain a quantitative description of the interactions that we are studying.

As the original results of this thesis, we answer several questions related to the repetition in parallel of these interactions. They follow the theme of looking at the optimality for Bob of treating different repetitions of an interaction independently when the interaction is repeated in parallel. That is, we consider the case in which Alice instantiates $n$ independent copies of her test: she follows exactly the same procedure in all of the $n$ parallel repetitions when determining what questions to send. The processes followed to determine the outcome of the interactions are completely independent as well.

There are several questions that one might ask  concerning what is Bob's optimal behaviour when a protocol is repeated several times in this way, depending on what does Bob want to optimize. In Chapter \ref{ch:avg}, we consider the setting in which a value is assigned by Bob to each of the outcomes, letting $v$ denote the best expected value that he can obtain as the outcome of an interaction. Formally speaking, and without a reference to any particular mathematical model for our interaction, this is the supremum over all possible processes by which Bob can operate of the expected value corresponding to the outcome of the protocol, when Bob follows that particular process. In both the classical and quantum models, the supremum will always be achieved, so that it may safely be replaced by the maximum.

Now, we consider the case in which Bob is trying to maximize the sum of the values obtained from $r$ repetitions of an interaction.  We ask then the question:

\begin{mylist}{\parindent}
\item[\ ] What is the optimum expected value per repetition that can be obtained for Bob when he considers all of the $n$ interactions?
\end{mylist}

One might think that given the fact that Alice is instantiating independently the copies of her test, the answer to this question is $v$, as this is the answer when Bob acts independently in the different repetitions. As we prove, this is indeed the correct answer.

In Chapter \ref{ch:nindep}, we look at the behaviour for Bob when he only cares about obtaining certain outcomes. Then, for a fixed choice of AliceÕs test and a particular choice of outcome, let $p$ denote the optimal probability for Bob of obtaining one of those outcomes. We identify these outcomes as the ``winning" outcomes. In the same way as in the definition of $v$, $p$ is more formally defined as the supremum of the probability that Bob achieves a winning outcome over all the processes by which Bob can operate. In the same way again as we have for $v$, in both our classical and quantum models this supremum can be safely replaced by a maximum.

When Bob is trying to optimize the average number of repetitions in which he obtains a winning outcome, the best he can do is to play independently his optimal strategy for achieving a winning outcome. This can be seen from assigning value $1$ to the winning outcomes and value $0$ to all other outcomes, and considering our result in Chapter \ref{ch:avg}. However, we can also consider the case in which Bob is not concerned with optimizing the average number of repetitions in which he obtains a winning outcome, but rather  with making sure that the number of repetitions in which he obtains a winning outcome is above a certain threshold. We ask then the question:

\begin{mylist}{\parindent}
\item[\ ] What is the optimum probability for Bob of achieving a winning outcome in at least $k$ of the  $n$ interactions?
\end{mylist}

Following the same reasoning as in the previous question, one might think that given the fact that Alice is instantiating independently the copies of her test, the answer to this question is $\sum_{ k \leq t \leq n}  {{n}\choose{t}} p^t (1-p)^{n-t}$. The reason is that this is the answer when Bob acts independently in the different repetitions. 

This is indeed the correct answer in the classical case. This can be proved from the observation that an optimal strategy for Bob in a classical model is always deterministic, and we will discuss later how it also follows as a special case of our semidefinite programming analysis in the quantum case. We will also see how this analysis can actually be extended to a few other cases, including the ones for the semidefinite programs in \cite{MolinaVW12}.

It is also known  that in the special case in which $k=n$, the answer to this question is indeed $\sum_{ k \leq t \leq n}   {{n}\choose{t}}p^t (1-p)^{n-t}$, which in this case equals $p^n$. 
In what is probably the most significant contribution in this thesis, we show how $\sum_{ k \leq t \leq n} {{n}\choose{t}} p^t (1-p)^{n-t}$ is in fact not in general the correct answer to this question. First, we show how the proof for the case in which $k=n$ fails to be generalized in a straightforward way in this case. Then, we give an explicit example of a test in which Bob can pass at least one of two repetitions with probability $1$, despite the fact that $p<1$. In our example, Bob's optimal probability of winning for a single repetition of the interaction is $\cos^2(\pi/8) \approx 0.85$. 
 
The ability of Bob to correlate his answers can be seen as a form of hedging, as we illustrate in a highly fictitious scenario.  In our scenario, Bob is offered the opportunity to take part in two potentially very lucrative but involving some risks games of chance, organized by Alice. These two games are completely identical to each other, and run independently. To earn the right to play in each of the games, Bob must contribute \$1 million of his own money, and he has an 85\% chance of winning if he plays optimally. For each game he wins, Bob receives a price of \$3 million, with a total \$2 million gain over his initial investment. If Bob does not win, he loses his \$1 million investment.  

Many people, if put in the place of Bob, would not hesitate to play both of the games, even taking out a \$2 million loan if necessary to do so. The expected gain from each of the games is \$1,550,000, and the only time that Bob loses money as an overall result is when Bob loses in both of the games. If we treat the games independently, the chance for a loss in both is 2.25\%. However, Bob could be a highly risk-averse person. He would greatly enjoy being a millionaire, but cannot or does not want to risk a 2.25\% chance of losing \$2 million. If the games run by Alice can be modelled classically, there is no way Bob can avoid this risk. However, if the two games have a model using quantum information with the same properties as the one in our example, Bob can be guaranteed to win in at least one of the games, and therefore obtain at least a total \$1 million gain. A choice of an appropriate quantum strategy allows Bob to hedge his bets perfectly.
  
There are other settings in which quantum effects that are not possible in the classical world have been discovered to be possible in an interaction between two parties that allows for quantum behaviour. However, our setting differs from some of the best-known such situations, such as the CHSH game \cite{ClauserH69} and the Mermin-Peres magic squares game \cite{Mermin90,Peres90}. In our setting, we do not have two parties collaborating to achieve a non-classical outcome. Instead, we have a prover-verifier setting, in which Bob is trying to convince Alice in order to achieve the winning outcome for an interaction.

In Chapter \ref{ch:ub}, we continue examining the same question as in Chapter \ref{ch:nindep}. As we said, we establish in Chapter \ref{ch:nindep} that it is not necessarily optimal for Bob to play independently when he is trying to win in at least a certain number of interactions. However, it still seems reasonable to think that how well Bob can do when he is trying to win in at least a certain number of repetitions should be somehow related to how well Bob can do when the interaction only occurs once, and he is trying to obtain the winning outcome. For example, it is clear that if Bob can make sure that he wins when the interaction only occurs once, then he is capable of making sure that he wins in at least a certain number of repetitions (since in fact, he can make sure that he wins in all repetitions). It is also possible to prove that if Bob does not have any chance of winning when the interaction only occurs once, then he does not have any chance of winning any number of interactions larger than zero when the game is played several times. This follows as a special case of our analysis for the quantum case. It can also be proved by contradiction starting from the observation that when only one repetition of the interaction is considered, Bob could simulate the setting in which several interactions are repeated in parallel, by simulating what would be the actions of Alice in the fictional copies of the interaction. 

It is then a reasonable aim to obtain general quantitative relations that express this idea. With this goal in mind, we try to upper bound the optimum probability for Bob of achieving the winning outcome in at least $k$ of the  $n$ interactions as a function of $p$. We will see how it is not hard to obtain from our formalization an upper bound of $\sum_{ k \leq t \leq n}  {{n}\choose{t}} p^t$. Using a more involved analysis, we obtain an improved upper bound of $p^k  {{n} \choose {k-1} }$.

In Chapter  \ref{ch:error}, we apply the results from the previous section to the study of error reduction for quantum interactive proof systems. These, generally speaking, are a particular case of the kind of interaction that we consider here. In this new situation, there is a \emph{string} $x$ known to Alice and Bob, which might or might not be a member of a \emph{language} $L$.

We also have an interaction of the form that we consider in our work, such that whenever $x \in L$ Bob can pass the test with probability at least $\alpha$, while whenever $x \notin L$ Bob can pass the test with probability at most $\beta < \alpha$.  Note that the fact that the value of $x$ is known to Alice and Bob implies that they can use this value to make decisions during their operation.

Assuming Bob is playing to maximize his chance of passing, Alice can then use the outcome of the test to make a guess about whether $x \in L$ or not. We can see that it is easy to make a guess that will be correct with high probability whenever $\alpha$ is close to $1$ and $\beta$ is close to $0$. Error reduction corresponds then to obtaining another test  with smaller $\beta$ and larger $\alpha$. In a natural conjecture for a possible way of reducing error, this new better test simply consists of a number of independent instantiations of the original test. The new test accepts if and only if some suitably chosen fraction of these independent tests (e.g. $\frac{\alpha+\beta}{2}$)  lead to Bob passing the test. This would improve on the more complicated strategy for reducing error in this situation that is known in the literature \cite{JainUW09}. 

If if was true that it is optimal for Bob to answer independently, that would easily prove the correctness of this natural strategy to reduce error. Indeed, under this assumption, the number of repetitions with a winning outcome when Bob plays optimally is described as  a binomial distribution parametrized by $p$ and $n$. Using the properties of the binomial distribution (e.g. using a Chernoff Bound), it would be then possible to prove that the probability that the new test produces a wrong guess about whether $x \in L$ decreases exponentially fast as a function of $n$.

Unfortunately, our results in Chapter  \ref{ch:nindep} shows that a proof method that uses the optimality of independent answers for Bob to prove the correctness of the natural strategy to reduce error would start with an incorrect assumption. On the other hand, maybe it is possible to prove the  correctness of the natural strategy to reduce error while replacing that incorrect assumption about Bob's optimal behaviour with a weaker one.  We will show how this is indeed the case for a limited range of values of $\alpha$ and $\beta$ (more exactly, whenever 
$\beta < 2^{-\frac{H(\alpha)}{\alpha}} < \alpha$), using our results from Chapter \ref{ch:ub}.

%% file: 02-background.tex
\newpage

\chapter{Background}\label{ch:back}

In this section we provide a summary of the mathematical background needed to develop the content of this thesis. Its main purpose is to unify the notation for the content of this thesis, and not to be completely exhaustive, but just to highlight concepts that might be less familiar to some readers. 

\section{Linear algebra}

We establish here the notation for linear algebra terms that will be used in this thesis. We assume familiarity with basic linear algebra concepts such as Hilbert spaces, positive semidefiniteness and tensor products. 
For any finite-dimensional complex Hilbert space $\X$ we write
$\lin{\X}$ to denote the set of linear operators acting on $\X$, we write $\mathbb{I}_{\X}$ to denote the identity
operator acting on $\X$,
we write $\herm{\X}$ to denote the set of Hermitian operators acting
on $\X$, we write $\pos{\X}$ to denote the set of positive
semidefinite operators acting on $\X$, and we write $\pd{\X}$ to denote the set of positive definite operators acting on $\X$. We write $\density{\X}$ to denote the set of density operators (positive semidefinite operators with unit trace) acting on $\X$. 

For Hermitian operators $A,B\in\herm{\X}$ the notations $A\geq B$ and
$B\leq A$ indicate that $A - B$ is positive semidefinite, and the
notations $A > B$ and $B < A$ indicate that $A - B$ is positive definite.

An inner product can be given to $\lin{\X}$, defined as  $\ip{A}{B} = \tr(A^{\ast}B)$. 
If $A,B\in\herm{\X}$, it holds that
$\ip{A}{B}$ is a real number and satisfies $\ip{A}{B} = \ip{B}{A}$.
For every choice of finite-dimensional complex Hilbert space $\X$ and
$\Y$, and for a given linear mapping of the form
$\Phi:\lin{\X}\rightarrow\lin{\Y}$, there is a unique mapping
$\Phi^{\ast}:\lin{\Y}\rightarrow\lin{\X}$ (known as the \emph{adjoint}
of $\Phi$) that satisfies
$\ip{Y}{\Phi(X)} = \ip{\Phi^{\ast}(Y)}{X}$ for all $X\in\lin{\X}$ and
$Y\in\lin{\Y}$.

We write the tensor product of a Hilbert space $\X$ with itself $n$ times, $\X \otimes \X \ldots \otimes \X$, as $\X^{\otimes n}$. We often consider several Hilbert spaces, all denoted by a common symbol (say $\X$), but with different subindices, corresponding to natural numbers in some range. We write $\X_{i \ldots j}$ to denote  the tensor product $\X_i \otimes \ldots  \otimes \X_j$  of a series of these spaces spanned by a sequence of consecutive subindices going from $i$ to $j$, inclusive.
 
During our exposition, we slightly abuse notation by identifying the tensor product of several Hilbert spaces with their tensor product in a different order. For example, we might write something like $\tr_{\Y_2} (P) = \mathbb{I}_{X_2} \otimes Q$, where $P \in \lin{\X_{1 \ldots 2} \otimes \Y_{1 \ldots 2}}$ and $Q \in \lin{\X_1 \otimes \Y_1}$.   We implicitly assume then that one of the sides in our equality is conjugated with an unitary operator that appropriately transposes the order of the spaces in the tensor product (for example, in this case we might conjugate $\tr_{\Y_2} (P)$ with the operator that sends $  x \otimes z \otimes y$ to $z \otimes  x \otimes y$ for all $x \in \X_1, y \in \Y_1, z \in X_2$ ).

We define now the Choi-Jamio{\l}kowski representation of a linear mapping from $\lin{\X}$ to $\lin{\Y}$. To do so, suppose $\op{dim}(\X) = n$ and assume that a standard orthonormal basis
$\{v_i : 1 \leq i \leq n \}$ of $\X$ has been selected. With respect to this basis, one defines the Choi-Jamio{\l}kowski
operator $J(\Phi)\in\lin{\Y\otimes\X}$ of a linear mapping
$\Phi:\lin{\X}\rightarrow\lin{\Y}$ as
\[
J(\Phi) = \sum_{1\leq i,j \leq n}
\Phi(v_iv_j^*) \otimes v_iv_j^*
\]
The mapping $J$ is a linear bijection from the space of mappings of
the form $\Phi:\lin{\X}\rightarrow\lin{\Y}$ to the operator space
$\lin{\Y\otimes\X}$.
It has the property that that $\Phi$ is completely positive if and only if
$J(\Phi) \in \pos{\Y\otimes\X}$, and that $\Phi$ is trace-preserving
if and only if $\tr_{\Y}(J(\Phi)) = \mathbb{I}_{\X}$
\cite{Choi75,Jamiolkowski72}.

Some properties of the elements in our formalism that we will use very often are the following:

\begin{mylist}{\parindent}

\item[ 1.] If $A \in \lin{\X_1 \otimes \Y_1}$  and $B \in \lin{\X_2 \otimes \Y_2}$, then $\tr_{\X_1 \otimes \X_2} (A \otimes B) = \tr_{\X_1} (A) \otimes \tr_{\X_2} (B)$

\item[ 2.] $A \geq B$ and $C \geq D$ implies $A \otimes C \geq B \otimes D$ for any choice of positive semidefinite operators $A, B, C$ and $D$.

\item[ 3.] If we have two maps $\Phi_1:\lin{\X_1}\rightarrow\lin{\Y_1}$ and  $\Phi_2:\lin{\X_2}\rightarrow\lin{\Y_2}$, then $J(\Phi_1 \otimes \Phi_2) = J(\Phi_1) \otimes J(\Phi_2)$.

\end{mylist}

\section{Quantum Information}

We introduce now some concepts concerning our mathematical modelling of quantum information processing. See \cite{NielsenC00} for a comprehensive introduction to a formal treatment of quantum information processing.

A \emph{register} is a hypothetical device that stores quantum
information.
Associated with a register $\reg{X}$ is a finite-dimensional complex
Hilbert space $\X$, and each quantum state of $\reg{X}$ is described
by a density operator $\rho\in\density{\X}$.
\emph{Qubits} are registers for which $\dim(\X) = 2$.
A \emph{measurement} of $\reg{X}$ is described by a set of positive
semidefinite operators $\{P_a\,:\,a\in\Sigma\}\subset\pos{\X}$,
indexed by a finite non-empty set of measurement outcomes $\Sigma$,
and satisfying the constraint $\sum_{a\in\Sigma}P_a = \mathbb{I}_{\X}$ (the
identity operator on $\X$).
If such a measurement is performed on $\reg{X}$ while it is in the
state $\rho$, each outcome $a\in\Sigma$ results with probability
$\ip{P_a}{\rho}$.

We can also consider information stored across several registers. If these registers are associated with finite-dimensional complex Hilbert spaces $\X_1 \ldots \X_n$, the finite-dimensional complex Hilbert space associated with their joint state is $\X_{1 \ldots n}$. Their joint state is then described by a density operator $\sigma \in \density{\X_{1 \ldots n}}$.

A \emph{quantum channel} is a completely positive and trace-preserving
linear mapping of the form \mbox{$\Phi:\lin{\X}\rightarrow\lin{\Y}$}. This
describes a hypothetical physical process that transforms each state
$\rho$ of a register $\reg{X}$ into the state $\Phi(\rho)$ of another
register $\reg{Y}$.
The set of all channels of this form is denoted $\channel{\X,\Y}$.
The identity channel that does nothing to a register $\reg{X}$ is
denoted $\mathbb{I}_{\lin{\X}}$.

\section{Semidefinite programming}

Semidefinite programming is an area of optimization which has been extensively used within quantum information theory (see for example \cite{CleveS07, JainJ10, LeeM11} for uses in quantum complexity theory, \cite{Rains01, NavascuesP08} for uses in the study of entanglement, and \cite{AudenaertD02,CerfF06} for uses in the study of quantum cloning).  More comprehensive discussions of semidefinite programming can be found in \cite{VandenbergheB96,Lovasz03,deKlerk02,BoydV04}, for instance.  We provide here the basic definitions and theorems used in our work.

\begin{definition}
A semidefinite program is specified by complex finite-dimensional Hilbert spaces $\X$ and $\Y$, and operators $\Phi$, $A$ and $B$, where:

\begin{mylist}{\parindent}
\item[1.] 
$\Phi: \lin{\X} \rightarrow \lin{\Y}$ is a Hermiticity-preserving
  linear mapping, and
\item[2.] $A\in\herm{\X}$ and $B\in\herm{\Y}$ are Hermitian operators,
\end{mylist}
for some choice of finite-dimensional complex Hilbert spaces $\X$ and $\Y$.

We associate with these operators two optimization problems, called the \emph{primal} and \emph{dual} problems:

\begin{center}
  \begin{minipage}{2.6in}
    \centerline{\underline{Primal problem}}\vspace{-7mm}
    \begin{align*}
      \text{maximize:}\quad & \ip{A}{X}\\
      \text{subject to:}\quad & \Phi(X) = B,\\
      & X\in\pos{\X}.
    \end{align*}
  \end{minipage}
  \hspace*{13mm}
  \begin{minipage}{2.6in}
    \centerline{\underline{Dual problem}}\vspace{-7mm}
    \begin{align*}
      \text{minimize:}\quad & \ip{B}{Y}\\
      \text{subject to:}\quad & \Phi^{\ast}(Y) \geq A,\\
      & Y\in\herm{\Y}.
    \end{align*}
  \end{minipage}
\end{center}
\noindent
The optimal primal value of this semidefinite program is
\[
\alpha = \sup\{\ip{A}{X}\,:\,X\in\pos{\X},\,\Phi(X) = B\}
\]
and the optimal dual value is
\[
\beta = \inf\{\ip{B}{Y}\,:\,Y\in\herm{\Y},\,\Phi^{\ast}(Y) \geq A\}.
\]
(It is to be understood that the supremum over an empty set is
$-\infty$ and the infimum over an empty set is $\infty$, so $\alpha$
and $\beta$ are well-defined values in the set
$\mathbb{R}\cup\{-\infty,\infty\}$.
In this thesis, however, we will only consider semidefinite programs for
which $\alpha$ and $\beta$ are finite).

\end{definition}

One of the most useful facts about a semidefinite program is that it always holds that $\alpha \leq \beta$. This is known as \emph{weak duality}.
The stronger condition $\alpha = \beta$, which is known as 
\emph{strong duality}, does not hold for every semidefinite program.
However, it is known that there are simple conditions under which it does hold.
The following theorem provides us with an example of such conditions:

\begin{theorem}[Slater's theorem for semidefinite programs]
  \label{theorem:Slater}
Let $(\Phi,A,B)$ be the operators in our definition of a semidefinite program, and let $\alpha$ and
$\beta$ be the optimal primal and dual values for the program.
\begin{mylist}{\parindent}
\item[1.]
  If the dual problem is feasible and there exists a positive definite operator
  $X\in\pd{\X}$ for which $\Phi(X) = B$,
  then $\alpha = \beta$ and there exists an operator $Y\in\herm{\Y}$
  such that $\Phi^{\ast}(Y)\geq A$ and $\ip{B}{Y} = \beta$.
\item[2.]
  If the primal problem is feasible and there exists a Hermitian operator
  $Y\in\herm{\Y}$ for which $\Phi^{\ast}(Y) > A$,
  then $\alpha = \beta$ and there exists a positive semidefinite 
  operator $X\in\pos{\X}$ such that $\Phi(X)=B$ and 
  $\ip{A}{X} = \alpha$.
\end{mylist}
\end{theorem}

This theorem states then in the first item that if there is a solution to the dual problem, as well as a positive definite solution to the primal problem, then strong duality holds, and an optimal dual solution is achievable. The second item gives us a similar condition, but reversing the role of the primal and dual problems. An $X$ such as the one in the first item and a $Y$ such as the one in the second item are called \emph{Slater points}.

%% file: 03-formalism.tex
\newpage

\chapter{Mathematical formalization} \label{ch:sdp}

We give now a presentation of the formalism that allows us to express the questions we ask in terms of semidefinite programs. This formalism was originally developed in \cite{GutoskiW07} and \cite{Gutoski10}. A related formalism for studying a similar kind of interaction was developed in \cite{ChiribellaDP09}.

The following definition formally defines an interaction of the kind that we described in the introduction. The interaction is assumed to have $r$ rounds (that is, $r$ questions from Alice to Bob) and $t$ different outcomes, which are indexed from $0$ to $t-1$:

\begin{definition} \label{def:game}

An interaction of the kind we study is defined by:

\begin{mylist}{\parindent}

\item[1. ] A series of $r$ quantum registers in which Alice writes her questions, which are then sent to Bob. The finite-dimensional complex Hilbert spaces  associated with these registers are denoted by $\X_1 \ldots \X_r$.

\item[2. ] A series of $r$ quantum registers registers in which Bob writes his answers, which are then sent to Alice. The finite-dimensional complex Hilbert spaces associated with these registers are denoted by $\Y_1 \ldots \Y_r$.

\item[3. ] A series of $r+1$ quantum registers registers that Alice uses to store her memory between the different points of the interaction.  The finite-dimensional complex Hilbert spaces associated with these registers are denoted by $\Z_1 \ldots \Z_{r+1}$.

\item [4. ] A quantum state that represents the first question sent by Alice to Bob, as well as the state of her initial memory. This state corresponds to a density matrix $\sigma \in \density{\X_1 \otimes \Z_1}$.

\item[5. ] A series of $r-1$ quantum channels  that correspond to the process by which Alice decides what question to ask. They produce a new question from Alice's memory and Bob's answer from the last question.  We denote them by $\Psi_2 \ldots \Psi_{r}$, with $\Psi_i $ sending elements of $\lin{\Y_{i-1} \otimes \Z_{i-1}}$ to elements of $\lin{\X_{i}}$.

\item[6. ] A quantum channel $\Psi_{r+1}$, sending elements of $\lin{\Y_{r} \otimes \Z_{r}}$ to elements of $\lin{\Z_{r+1}}$, the memory space for Alice after she receives the last answer.

\item[7. ] A projective measurement  $\{Q_i : 0 \leq i \leq k-1\}$ by which Alice decides the outcome of the interaction. This measurement is performed on $\Z_{r+1}$. The positive semidefinite operator corresponding to outcome $k$ is given by $Q_k$.

\end{mylist}

\end{definition}

Note that elements 6 and 7 could be merged together in a POVM measurement, but it will be more convenient for a later discussion to present the interaction in this way.

We formally define now the processes by which Bob can operate:

\begin{definition} \label{def:bob}

A possible process by which Bob can operate is given by: 
\begin{mylist}{\parindent}

\item[1. ]  A series of $r-1$ quantum registers registers that Bob uses to store his memory between sending an answer to Alice and receiving the next question. The finite-dimensional complex Hilbert spaces associated with these registers are denoted by $\W_1 \ldots \W_{r-1}$.

\item[2. ]  A series of $r$ quantum channels that correspond to the process by which Bob decides his answers. They produce an answer from Alice's question and Bob's memory. We denote them by $\Gamma_1 \ldots \Gamma_r$. If $r=1$, $\Gamma_1$ sends elements of $\lin{\X_1}$ to elements of $\lin{\Y_1}$. If $r>1$,  $\Gamma_1$ sends elements of $\lin{\X_1}$ to elements of $\lin{\Y_1 \otimes \W_1}$, $\Gamma_i$ for $1 < i < r$ sends elements of $\lin{\X_i \otimes \W_{i-1}}$ to elements of $\lin{\Y_r \otimes \W_r}$, and $\Gamma_r$ sends elements of  $\lin{\X_r \otimes \W_{r-1}}$ to elements of $\lin{\Y_r}$.

\end{mylist}

\end{definition}

A graphical representation of these definitions can be seen in Figure \ref{fig:interaction} 

\begin{figure}[t]
  \begin{center}
  	\input{figinteraction.tex}
  \end{center}
  \caption{The interactions between Alice and Bob that we study}
  \label{fig:interaction}
\end{figure}
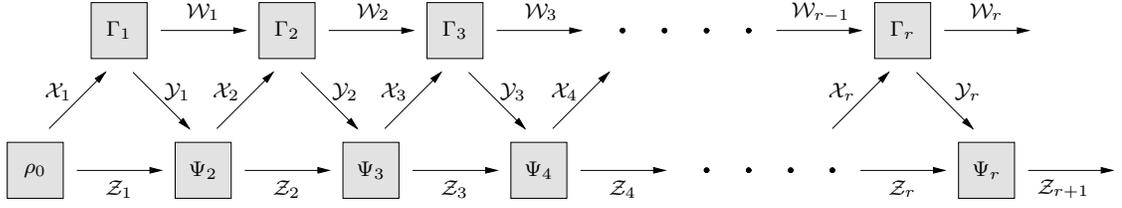

The main result that allows us to express our questions in terms of semidefinite programs, which originally appeared in \cite{Gutoski10}, is the following one:

\begin{lemma} \label{lm:rep}
There exists operators $P_0,  \ldots, P_{t-1} \in \pos{ \Y_{1 \ldots r} \otimes \X_{1 \ldots r} }$ , and a  map from the set of  possible processes by which Bob can operate to $\pos{\Y_{1 \ldots r} \otimes \X_{1 \ldots r}}$, such that the probability that outcome $i$ is obtained by the behaviour of Bob corresponding to $X \in \pos{ \Y_{1 \ldots r} \otimes \X_{1 \ldots r} }$ is $\ip{P_i}{X}$. 

Moreover,

\begin{mylist}{\parindent}

\item[1. ]

The subset of $\pos{\Y_{1 \ldots r} \otimes \X_{1 \ldots r}}$ to which the possible strategies for Bob are sent is composed by the elements $X$ such that there exist $X_1, \ldots, X_{r-1}$, with $\X_i \in \pos{\Y_{1 \ldots i} \otimes \X_{1 \ldots i}}$,  satisfying:

\begin{center}
  \begin{minipage}{2.6in}
    \begin{align*}
 \quad&        \tr_{\Y_1}(X_1) = \mathbb{I}_{\X_1}\\
  \quad&   \tr_{\Y_2}(X_2) -X_1 \otimes \mathbb{I}_{\X_2} = 0\\
   \quad & \vdots \\
 \quad &     \tr_{\Y_r}(X) -X_{r-1} \otimes \mathbb{I}_{\X_r}  =  0\\    
 \end{align*}
  \end{minipage}
\end{center}

\item[2. ] The operators $P_0,  \ldots, P_{t-1}$ are such that there exist $\{P_i^j \in \pos{\Y_{1 \ldots j}  \otimes \X_{1 \ldots j} }\ \ : \ 0 \leq i \leq t-1, 1 \leq j \leq r-1\}$, $\rho \in \density{\X_1}$, and $\{R_j \in \pos{\Y_{1 \ldots {j-1}} \otimes \X_{1 \ldots j}} : 2 \leq j \leq r \}$, satisfying:

\begin{center}
  \begin{minipage}{2.6in}
    \begin{align*}
 \quad &  \sum_{0 \leq i \leq t-1} P_i^1 = \mathbb{I}_{\Y_1} \otimes \rho \\
 \quad &  \tr_{\X_2}(R_2) =    \sum_{0 \leq i \leq t-1} P_i^1 \\
 \quad &  \sum_{0 \leq i \leq t-1} P_i^2 = \mathbb{I}_{\Y_2} \otimes R_2 \\
  \quad &      \vdots \\
 \quad &  \tr_{\X_r}(R_r) =    \sum_{0 \leq i \leq t-1} P_i^{r-1} \\
 \quad &  \sum_{0 \leq i \leq t-1} P_i = \mathbb{I}_{\Y_r} \otimes R_r \\
 \end{align*}
  \end{minipage}
\end{center}

\end{mylist}

\end{lemma}

This will allow us to express the questions we examine in the next chapters as questions about semidefinite programs.  We will not give a proof of this Lemma here. However, we will mention the main ideas behind it. This will motivate how does the formalism apply to the case in which an interaction is repeated several times in parallel, with Alice acting independently. The main idea to prove this Lemma consists of looking at all the actions of Alice together as an operator $\Xi_A$ from $\lin{ \Y_1 \otimes \ldots \otimes \Y_{r} }$ to $\lin { \X_1 \otimes \ldots \otimes \X_r \otimes \Z_{r+1}}$, followed by a projective measurement $\{Q_k\}$ of   $\Z_{r+1}$, and tracing out  $\Z_{r+1}$. Similarly, we look at the actions of Bob as an operator $\Xi_B$ from $\lin{ \X_1 \otimes \ldots \X_{r} }$ to $\lin { \Y_1 \otimes \ldots \otimes \Y_r}$, as can be seen in figure \ref{fig:strategy}.Then,  $J(\Xi_B)$ is the object to which the actions of Bob map in the previous lemma, and $P_k$ is given by $J\left( \left( \tr_{\Z_{r+1}} \left(Q_k \Xi_A \right) \right)^* \right)$. The restrictions on such operators that appear in Lemma \ref{lm:rep} are then the conditions that $J\left( \left( \tr_{\Z_{r+1}} \left(Q_k \Xi_A \right) \right)^* \right)$and $J(\Xi_B)$ satisfy for  operations $\Xi_A$ and $\Xi_B$ following the causal structure in our description of an interaction, as it is proved in \cite{GutoskiW07} and \cite{Gutoski10}, with the operators $X_i$ and $R_j$ corresponding to the first rounds of the interaction. 

\begin{figure}[t]
  \begin{center}
  	\input{figstrategy.tex}
  \end{center}
  \caption{The map from $\lin{ \X_1 \otimes \ldots \X_{r} }$ to $\lin { \Y_1 \otimes \ldots \otimes \Y_r} $ corresponding to the actions of Bob, in the particular case $r=3$}
  \label{fig:strategy}
\end{figure}

Now, when we consider parallel repetition with Alice acting independently, we have that the new operator for Alice is given by 

\begin{eqnarray*}
J\left( \left( \tr_{\Z_{r+1}^{\otimes n}} \left(Q_{i_1} Q_{i_2} \ldots Q_{i_n} \Xi_A^{\otimes n} \right) \right)^* \right)  & = &   J\left( \left( \tr_{\Z_{r+1}} \left(Q_{i_1} \Xi_A \right) \right)^* \otimes  \ldots \otimes   \left( \tr_{\Z_{r+1}} \left(Q_{i_n} \Xi_A \right) \right)^* \right) \\ 
&  =  & J\left( \left( \tr_{\Z_{r+1}} \left(Q_{i_1} \Xi_A \right) \right)^* \right) \otimes  \ldots \otimes J \left(  \left( \tr_{\Z_{r+1}} \left(Q_{i_n} \Xi_A \right) \right)^* \right)  \\
&  =  & P_{i_1} \otimes \ldots \otimes P_{i_n},
\end{eqnarray*}

\noindent using the properties of tensor products mentioned in Chapter \ref{ch:back}. This motivates then the following Lemma:

\begin{lemma}  \label{cr:par}
Consider the case in which an interaction of the kind we consider here is repeated in parallel $n$ times. Then, if we let $P_0,  \ldots, P_{t-1} $ be the operators in the previous lemma, now there is a map from the set of possible strategies for Bob to  $\pos{ \left(\Y_{1 \ldots r} \otimes \X_{1 \ldots r} \right)^{\otimes n}}$ such that the probability that outcomes $\{i_j : j \in 1 \ldots n\}$ are obtained in the $n$ different repetitions when the behaviour of Bob corresponds to $X$ is given by $\ip{ P_{i_1} \otimes P_{i_2} \otimes \ldots \otimes P_{i_n} }{X}$. Moreover, the subset of $\pos{\Y_{1 \ldots r} \otimes \X_{1 \ldots r}}$ to which the possible strategies for Bob are sent is composed by the elements $X$ such that there exist $X_1, \ldots, X_{r-1}$, with $X_i \in \pos{\Y_{1 \ldots i} ^{\otimes n} \otimes \X_{1 \ldots i}^{\otimes n} }$,  satisfying:

\begin{center}
  \begin{minipage}{2.6in}
    \begin{align*}
  \quad &   \tr_{\Y_1^{\otimes n}}(X_1)  =  \mathbb{I}_{\X_1^{\otimes n}}\\
  \quad &    \tr_{\Y_2^{\otimes n}}(X_2) -X_1 \otimes \mathbb{I}_{\X_2^{\otimes n}} = 0\\
   \quad &   \vdots \\
   \quad &   \tr_{\Y_r^{\otimes n}}(X) -X_{r-1} \otimes \mathbb{I}_{\X_r^{\otimes n}} = 0\\    
 \end{align*}
  \end{minipage}
\end{center}

\end{lemma}

To model the situation when only classical information is being exchanged inside this more general model, it is enough to assume that Alice appends a completely dephasing channel before and after each of her actions, and her measurement corresponds as well to a set of diagonal matrices. The completely dephasing channel just makes equal to zero all non-diagonal elements from the density matrix representing its input state, and leaves the diagonal elements unmodified, making sure then that only classical information is being sent. In this case, it is not hard to see that $J\left( \left( \tr_{\Z_{r+1}} \left(Q_k \Xi_A \right) \right)^* \right)$ will be a diagonal matrix, since for the terms of the Choi-Jamio{\l}kowski  representation corresponding to a non-diagonal input we have that the channel  $\Xi_A$ will map them to 0, and for all the other ones, we have that they are mapped to a classical state, that is, a diagonal matrix. $\rho$ and the $R_i$ are diagonal as well, since they can be given a similar interpretation in terms of the Choi-Jamio{\l}kowski  representations for the actions of Alice when we ignore her measurement and shorten the length of the interaction, as it is shown in \cite{GutoskiW07} and \cite{Gutoski10}.

%% file: figinteraction.tex
    \setlength{\unitlength}{0.000245in}
    \begin{picture}(23734,4539)(-600,200)
      \scriptsize
      \texture{8101010 10000000 444444 44000000 11101 11000000 444444 44000000
	101010 10000000 444444 44000000 10101 1000000 444444 44000000
	101010 10000000 444444 44000000 11101 11000000 444444 44000000
	101010 10000000 444444 44000000 10101 1000000 444444 44000000 }
	
      \shade\path(20412,1212)(21612,1212)(21612,12)(20412,12)(20412,1212)
      \shade\path(18612,4212)(19812,4212)(19812,3012)(18612,3012)(18612,4212)
      \shade\path(1812,4212)(3012,4212)(3012,3012)(1812,3012)(1812,4212)
      \shade\path(3612,1212)(4812,1212)(4812,12)(3612,12)(3612,1212)
      \shade\path(5412,4212)(6612,4212)(6612,3012)(5412,3012)(5412,4212)
      \shade\path(7212,1212)(8412,1212)(8412,12)(7212,12)(7212,1212)
      \shade\path(9012,4212)(10212,4212)(10212,3012)(9012,3012)(9012,4212)
      \shade\path(10812,1212)(12012,1212)(12012,12)(10812,12)(10812,1212)

      \shade\path(12,1212)(1212,1212)(1212,12)(12,12)(12,1212)

      \put(612,612){\makebox(0,0){$\rho_0$}}

      \put(4212,612){\makebox(0,0){$\Psi_2$}}
      \put(7812,612){\makebox(0,0){$\Psi_3$}}
      \put(11412,612){\makebox(0,0){$\Psi_4$}}
      \put(21012,612){\makebox(0,0){$\Psi_r$}}
      \put(2412,3612){\makebox(0,0){$\Gamma_1$}}
      \put(6012,3612){\makebox(0,0){$\Gamma_2$}}
      \put(9612,3612){\makebox(0,0){$\Gamma_3$}}
      \put(19212,3612){\makebox(0,0){$\Gamma_r$}}


      \path(16512,3612)(18312,3612)
      \blacken\path(18072.000,3552.000)(18312.000,3612.000)(18072.000,3672.000)
      (18072.000,3552.000)

      \path(18312,612)(20112,612)
      \blacken\path(19872.000,552.000)(20112.000,612.000)(19872.000,672.000)
      (19872.000,552.000)

      \path(17712,1512)(18912,2712)
      \blacken\path(18784.721,2499.868)(18912.000,2712.000)(18699.868,2584.721)
      (18784.721,2499.868)

      \path(19512,2712)(20712,1512)
      \blacken\path(20499.868,1639.279)(20712.000,1512.000)(20584.721,1724.132)
      (20499.868,1639.279)

      \path(912,1512)(2112,2712)
      \blacken\path(1984.721,2499.868)(2112.000,2712.000)(1899.868,2584.721)
      (1984.721,2499.868)

      \path(2712,2712)(3912,1512)
      \blacken\path(3699.868,1639.279)(3912.000,1512.000)(3784.721,1724.132)
      (3699.868,1639.279)

      \path(4512,1512)(5712,2712)
      \blacken\path(5584.721,2499.868)(5712.000,2712.000)(5499.868,2584.721)
      (5584.721,2499.868)

      \path(6312,2712)(7512,1512)
      \blacken\path(7299.868,1639.279)(7512.000,1512.000)(7384.721,1724.132)
      (7299.868,1639.279)

      \path(8112,1512)(9312,2712)
      \blacken\path(9184.721,2499.868)(9312.000,2712.000)(9099.868,2584.721)
      (9184.721,2499.868)

      \path(9912,2712)(11112,1512)
      \blacken\path(10899.868,1639.279)(11112.000,1512.000)(10984.721,1724.132)
      (10899.868,1639.279)


      \path(3312,3612)(5112,3612)
      \blacken\path(4872.000,3552.000)(5112.000,3612.000)(4872.000,3672.000)
      (4872.000,3552.000)

      \path(6912,3612)(8712,3612)
      \blacken\path(8472.000,3552.000)(8712.000,3612.000)(8472.000,3672.000)
      (8472.000,3552.000)

      \path(1512,612)(3312,612)
      \blacken\path(3072.000,552.000)(3312.000,612.000)(3072.000,672.000)
      (3072.000,552.000)

      \path(5112,612)(6912,612)
      \blacken\path(6672.000,552.000)(6912.000,612.000)(6672.000,672.000)
      (6672.000,552.000)

      \path(8712,612)(10512,612)
      \blacken\path(10272.000,552.000)(10512.000,612.000)(10272.000,672.000)
      (10272.000,552.000)

      \path(10512,3612)(12312,3612)
      \blacken\path(12072.000,3552.000)(12312.000,3612.000)(12072.000,3672.000)
      (12072.000,3552.000)

      \path(12312,612)(14112,612)
      \blacken\path(13872.000,552.000)(14112.000,612.000)(13872.000,672.000)
      (13872.000,552.000)

      \path(11712,1512)(12912,2712)
      \blacken\path(12784.721,2499.868)(12912.000,2712.000)(12699.868,2584.721)
      (12784.721,2499.868)

      \path(20112,3612)(21912,3612)
      \blacken\path(21672.000,3552.000)(21912.000,3612.000)(21672.000,3672.000)
      (21672.000,3552.000)

      \path(21912,612)(23712,612)
      \blacken\path(23472.000,552.000)(23712.000,612.000)(23472.000,672.000)
      (23472.000,552.000)

      \put(16812,612){\blacken\ellipse{100}{100}}
      \put(16812,612){\ellipse{100}{100}}
      \put(15012,3612){\blacken\ellipse{100}{100}}
      \put(15012,3612){\ellipse{100}{100}}
      \put(13212,3612){\blacken\ellipse{100}{100}}
      \put(13212,3612){\ellipse{100}{100}}
      \put(14112,3612){\blacken\ellipse{100}{100}}
      \put(14112,3612){\ellipse{100}{100}}
      \put(15912,3612){\blacken\ellipse{100}{100}}
      \put(15912,3612){\ellipse{100}{100}}
      \put(15012,612){\blacken\ellipse{100}{100}}
      \put(15012,612){\ellipse{100}{100}}
      \put(15912,612){\blacken\ellipse{100}{100}}
      \put(15912,612){\ellipse{100}{100}}
      \put(17712,612){\blacken\ellipse{100}{100}}
      \put(17712,612){\ellipse{100}{100}}

      \put(1112,2112){\makebox(0,0)[b]{$\X_1$}}
      \put(4712,2112){\makebox(0,0)[b]{$\X_2$}}
      \put(8312,2112){\makebox(0,0)[b]{$\X_3$}}
      \put(11962,2112){\makebox(0,0)[b]{$\X_4$}}
      \put(17912,2112){\makebox(0,0)[b]{$\X_r$}}

      \put(3662,2112){\makebox(0,0)[b]{$\Y_1$}}
      \put(7262,2112){\makebox(0,0)[b]{$\Y_2$}}
      \put(10862,2112){\makebox(0,0)[b]{$\Y_3$}}
      \put(20612,2112){\makebox(0,0)[b]{$\Y_r$}}

      \put(4212,4012){\makebox(0,0){$\W_1$}}
      \put(7862,4012){\makebox(0,0){$\W_2$}}
      \put(11462,4012){\makebox(0,0){$\W_3$}}
      \put(17462,4012){\makebox(0,0){$\W_{r-1}$}}

      \put(2412,212){\makebox(0,0){$\Z_1$}}
      \put(6012,212){\makebox(0,0){$\Z_2$}}
      \put(9612,212){\makebox(0,0){$\Z_3$}}
      \put(13212,212){\makebox(0,0){$\Z_4$}}
      \put(19212,212){\makebox(0,0){$\Z_{r}$}}
      \put(21012,4012){\makebox(0,0){$\W_r$}}
      \put(22662,212){\makebox(0,0){$\Z_{r+1}$}}

    \end{picture}

%% file: figstrategy.tex
  \setlength{\unitlength}{0.000445in}
  \begin{picture}(6999,2697)(0,-10)
    \scriptsize

    \texture{8101010 10000000 444444 44000000 11101 11000000 444444 44000000
	101010 10000000 444444 44000000 10101 1000000 444444 44000000
	101010 10000000 444444 44000000 11101 11000000 444444 44000000
	101010 10000000 444444 44000000 10101 1000000 444444 44000000 }

    \shade\path(1587,1137)(2262,1137)(2262,12)(1587,12)(1587,1137)
    \shade\path(3162,1812)(3837,1812)(3837,687)(3162,687)(3162,1812)
    \shade\path(4737,2487)(5412,2487)(5412,1362)(4737,1362)(4737,2487)

    \put(1925,575){\makebox(0,0){$\Gamma_1$}}
    \put(3500,1250){\makebox(0,0){$\Gamma_2$}}
    \put(5075,1925){\makebox(0,0){$\Gamma_3$}}


    \put(687,1587){\makebox(0,0)[r]
      {$\xi\left\{\rule[9mm]{0mm}{0mm}\right.$}}
    \put(6312,912){\makebox(0,0)[l]
      {$\left.\rule[9mm]{0mm}{0mm}\right\}\Xi(\xi)$}}


    \path(2362,912)(3062,912)
    \blacken\path(2942.000,882.000)(3062.000,912.000)
      (2942.000,942.000)(2942.000,882.000)


    \path(787,912)(1487,912)
    \blacken\path(1367.000,882.000)(1487.000,912.000)
      (1367.000,942.000)(1367.000,882.000)

    \path(3937,1587)(4637,1587)
    \blacken\path(4517.000,1557.000)(4637.000,1587.000)
      (4517.000,1617.000)(4517.000,1557.000)


    \path(5512,1587)(6212,1587)
    \blacken\path(6092.000,1557.000)(6212.000,1587.000)
      (6092.000,1617.000)(6092.000,1557.000)

    \path(2362,237)(6212,237)
    \blacken\path(6092.000,207.000)(6212.000,237.000)
      (6092.000,267.000)(6092.000,207.000)

    \path(787,1587)(3062,1587)
    \blacken\path(2942.000,1557.000)(3062.000,1587.000)
      (2942.000,1617.000)(2942.000,1557.000)

    \path(3937,912)(6212,912)
    \blacken\path(6092.000,882.000)(6212.000,912.000)
      (6092.000,942.000)(6092.000,882.000)

    \path(787,2262)(4637,2262)
    \blacken\path(4517.000,2232.000)(4637.000,2262.000)
      (4517.000,2292.000)(4517.000,2232.000)

    \put(5862,1812){\makebox(0,0){$Y_3$}}
    \put(5862,1137){\makebox(0,0){$Y_2$}}
    \put(5862,462){\makebox(0,0){$Y_1$}}
    \put(1137,2487){\makebox(0,0){$X_3$}}
    \put(1137,1812){\makebox(0,0){$X_2$}}
    \put(1137,1137){\makebox(0,0){$X_1$}}
    \put(4287,1812){\makebox(0,0){$Z_2$}}
    \put(2712,1137){\makebox(0,0){$Z_1$}}
  \end{picture}

%% file: 04-average.tex
\newpage

\chapter{Optimal expected value under parallel repetition} \label{ch:avg}

We study in this chapter the situation where  each outcome $i \in \{0 \ldots t-1\}$ of an interaction is associated with a value $v_i$, and Bob is concerned with maximizing the expected value that he obtains from the interaction. We prove that when an interaction is instantiated several times in parallel, it is optimal for Bob to act independently when is he trying to maximize the expected value per instantiation.

To express this situation formally, consider a fixed interaction, described using our formalism in Chapter \ref{ch:sdp}, and a fixed process by which Bob can operate, represented by $X \in \pos{ \Y_{1 \ldots r} \otimes {\X_{1 \ldots r} } }$. As the probability that outcome $i$ is obtained is $\ip{P_i}{X}$, the expected value obtained by Bob for an interaction when following a process represented by $X$ is $\sum_i v_i \ip{P_i}{X}$.  Therefore, this is the objective function that Bob is trying to maximize in the situation studied in this chapter. We obtain then that the following optimization problem corresponds to the problem of finding the optimal strategy for Bob in the setting where he is trying to maximize the expected value of the outcome:

\begin{center}
  \begin{minipage}{2.6in}
    \centerline{\underline{Primal Problem 1}}\vspace{-7mm}
    \begin{align*}
      \text{maximize:} & \sum_i v_i \ip{P_i}{X}\\
      \text{subject to:}    \quad & \tr_{\Y_1}(X_1) = \mathbb{I}_{\X_1}\\
     \quad &  \tr_{\Y_2}(X_2) -X_1 \otimes \mathbb{I}_{\X_2} = 0\\
    \quad &  \vdots \\
      \quad &\tr_{\Y_r}(X) -X_{r-1} \otimes \mathbb{I}_{\X_r}  = 0\\    
	& X\in\pos{\Y_{1 \ldots r} \otimes \X_{1 \ldots r} }
	, X_i \in\pos{\Y_{1 \ldots i} \otimes \X_{1 \ldots i} }
    \end{align*}
  \end{minipage}
\end{center}

To compute the dual, we express this is in the exact same form that appears in our definition of semidefinite program:

\begin{center}
  \begin{minipage}{2.6in}
    \centerline{\underline{Primal Problem 2}}\vspace{-7mm}
    \begin{align*}
      \text{maximize:} & \ip{ \begin{pmatrix}
      0 & \ & \ & 0  \\
      \  & 0 & \ & \  \\
      \  & \ & \ddots & \  \\
     0 & \ & \  &  \sum_i v_i P_i   \\
    \end{pmatrix} }{\begin{pmatrix}
      X_1 & \ & \ & \  \\
      \  & X_2 & \ & \  \\
      \  & \ & \ddots & \  \\
      \  & \ & \  & X  \\
    \end{pmatrix} }\\
      \text{subject to:} & {\begin{pmatrix}
     \tr_{\Y_1}(X_1) & \ & \ & 0  \\
      \  &  \tr_{\Y_2}(X_2) -X_1 \otimes \mathbb{I}_{\X_2} & \ & \  \\
      \  & \ & \ddots & \  \\
      0  & \ & \  & \tr_{\Y_r}(X) -X_{r-1} \otimes \mathbb{I}_{\X_r}  \\
    \end{pmatrix} }  =   {\begin{pmatrix}
     \mathbb{I}_{\X_1} & \ & \ & 0 \\
      \  & 0 & \ & \  \\
      \  & \ & \ddots & \  \\
      0  & \ & \  & 0  \\
    \end{pmatrix} }\\
     \quad &   \begin{pmatrix}
      X_1 & \ & \ & \  \\
      \  & X_2 & \ & \  \\
      \  & \ & \ddots & \  \\
      \  & \ & \  & X  \\
    \end{pmatrix} \in \pos{ (\Y_1 \otimes \X_1) \oplus  ( \Y_{1 \ldots 2}  \otimes \X_{1 \dots 2} )  \ldots \oplus  (\Y_{1 \ldots r}  \otimes \X_{1 \dots r} )  }\\
    \end{align*}
  \end{minipage}
\end{center}

Note that we are ignoring the non-diagonal blocks of the solution. The reason why we can ignore them is that they are ignored by both the function that we are trying to optimize and the constraint $\Phi$ of the semidefinite program. Their presence does not alter either the set of $X, \{X_i\}$ that represent feasible solutions, since for any feasible solution to this new problem, its blocks along the diagonal will be positive semidefinite, and for any positive semidefinite $X, \{X_i\}$, leaving the non-diagonal blocks as zero will give us a positive semidefinite matrix.

As $\Phi$ is represented by the action

\[\begin{pmatrix}
      X_1 & \ & \ & \  \\
      \  & X_2 & \ & \  \\
      \  & \ & \ddots & \  \\
      \  & \ & \  & X  \\
    \end{pmatrix} \rightarrow \begin{pmatrix}
     \tr_{\Y_1}(X_1) & \ & \ & 0  \\
      \  &  \tr_{\Y_2}(X_2) -X_1 \otimes \mathbb{I}_{\X_2} & \ & \  \\
      \  & \ & \ddots & \  \\
      0  & \ & \  & \tr_{\Y_r}(X) -X_{r-1} \otimes \mathbb{I}_{\X_r}  \\
    \end{pmatrix},
    \]
    
   \noindent its adjoint $\Phi^*$ corresponds to the action
    
    \[ \begin{pmatrix}
      Y_1 & \ & \ & \  \\
      \  & Y_2 & \ & \  \\
      \  & \ & \ddots & \  \\
      \  & \ & \  & Y  \\
    \end{pmatrix}  \rightarrow   \begin{pmatrix}
      Y_1  \otimes \mathbb{I}_{\Y_1} - \tr_{\X_2}(Y_2)& \ & \ & \  \\
      \  &  Y_2  \otimes \mathbb{I}_{\Y_2} - \tr_{\X_3}(Y_3) & \ & \  \\
      \  & \ & \ddots & \  \\
      \  & \ & \  & Y \otimes \mathbb{I}_{\Y_r}  \\
    \end{pmatrix}, \] \noindent as can be verified by a simple computation.

The dual problem for our situation will correspond then to:

\begin{center}
  \begin{minipage}{2.6in}
    \centerline{\underline{Dual Problem 1}}\vspace{-7mm}
    \begin{align*}
      \text{minimize:} & \ip{ \begin{pmatrix}
     \mathbb{I}_{\X_1} & \ & \ & 0 \\
      \  & 0 & \ & \  \\
      \  & \ & \ddots & \  \\
      0  & \ & \  & 0  \\
    \end{pmatrix} }{ \begin{pmatrix}
      Y & \ & \ & \  \\
      \  & Y_2 & \ & \  \\
      \  & \ & \ddots & \  \\
      \  & \ & \  & Y_r  \\
    \end{pmatrix} } \\
      \text{subject to:} &  \begin{pmatrix}
      Y  \otimes \mathbb{I}_{\Y_1} - \tr_{\X_2}(Y_2)& \ & \ & \  \\
      \  &  Y_2  \otimes \mathbb{I}_{\Y_2} - \tr_{\X_3}(Y_3) & \ & \  \\
      \  & \ & \ddots & \  \\
      \  & \ & \  & Y_r \otimes \mathbb{I}_{\Y_r}  \\
    \end{pmatrix}  \geq   \begin{pmatrix}
      0 & \ & \ & 0  \\
      \  & 0 & \ & \  \\
      \  & \ & \ddots & \  \\
     0 & \ & \  &  \sum_i v_i P_i   \\     \end{pmatrix} \\     
      \quad &
    \begin{pmatrix}
      Y & \ & \ & \  \\
      \  & Y_2 & \ & \  \\
      \  & \ & \ddots & \  \\
      \  & \ & \  & Y_r  \\
    \end{pmatrix}  \in \herm{ \X_1 \oplus  ( \Y_1 \otimes X_{1 \dots 2} )  \ldots \oplus  ( \Y_{1 \ldots r-1}  \oplus \X_{1 \dots r} )  }
    \end{align*}
  \end{minipage}
\end{center}

We can simplify this and write it as:

\begin{center}
  \begin{minipage}{2.6in}
    \centerline{\underline{Dual Problem 2}}\vspace{-7mm}
    \begin{align*}
      \text{minimize:} & \tr(Y)  \\
      \text{subject to:}  \quad &  Y  \otimes \mathbb{I}_{\Y_1} - \tr_{\X_2}(Y_2) \geq 0 \\
      \quad &   Y_2  \otimes \mathbb{I}_{\Y_2} - \tr_{\X_3}(Y_3)  \geq 0 \\
      \vdots  \\
       \quad &   Y_r \otimes \mathbb{I}_{\Y_r} \geq  \sum_i v_i P_i \\ 
       	& Y \in\herm{\X_1 }
	, Y_i \in\herm{\Y_{1 \ldots i-1} \otimes \X_{1 \ldots i} }
    \end{align*}
  \end{minipage}
\end{center}

\noindent Note that even if we only  explicitly require $ Y \in\herm{\X_1, Y_i}, Y_i \in\herm{\Y_{1 \ldots i-1} \otimes \X_{1 \ldots i} }$, it must actually be the case that   $ Y \in\pos{\X_1 }, Y_i \in\pos{\Y_{1 \ldots i-1} \otimes \X_{1 \ldots i} }$. In the case of $Y_r$, this is because the fact that the last constraint is satisfied implies that $Y_r$ is $\geq$ than a positive semidefinite operator. Feasibility for the second to last constraint implies then that $Y_{r-1}$ is positive semidefinite, and so on. This will actually be the case for all the dual problems that consider in this thesis.
		
Note as well that there is a case $r=1$  that is somewhat of a notational corner case in our semidefinite programming formulations. In this situation, we have $X_1=X$ and $Y_r=Y$, with our semidefinite programs being then:

\begin{center}
  \begin{minipage}{2.6in}
    \centerline{\underline{Primal problem}}\vspace{-7mm}
    \begin{align*}
      \text{maximize:}\quad & \ip{P_1}{X}\\
      \text{subject to:}\quad & \tr_{\Y}(X) = \mathbb{I}_{\X_1},\\
      & X\in\pos{\Y_1\otimes\X_1}.
    \end{align*}
  \end{minipage}
  \hspace*{13mm}
  \begin{minipage}{2.6in}
    \centerline{\underline{Dual problem}}\vspace{-7mm}
    \begin{align*}
      \text{minimize:}\quad & \tr(Y)\\
      \text{subject to:}\quad & \mathbb{I}_{\Y_1} \otimes Y \geq P_1,\\
      & Y\in\herm{\X_1}.
    \end{align*}
  \end{minipage}
\end{center}

\noindent All of the proofs and derivations of semidefinite programs in the rest of this thesis can be adapted to this case without problems.

We prove now that both statements in Theorem \ref{theorem:Slater} apply, so there are optimal primal and dual solutions to our semidefinite program, and they have the same value. Indeed,

\begin{mylist} {\parindent}
\item[1. ] For the form of the primal problem that follows the definition of semidefinite program (Primal Problem 2), we have that there is a positive definition solution, which can be obtained by letting each of the diagonal blocks of our solution be an appropriate multiple of the identity. For example, we can have:

\[
X_1 = \frac{\mathbb{I}_{\Y_1 \otimes \X_1} }{\dim{Y_1}}, \ X_2 = \frac{\mathbb{I}_{\Y_{1\ldots2} \otimes \X_{1\ldots 2}} }{\dim{Y_1}*\dim{Y_2}},   \ldots, X = \frac{\mathbb{I}_{\Y_{1\ldots r} \otimes \X_{1\ldots r }} }{\dim{Y_1}*\dim{Y_2}*\ldots\dim{Y_r}}
\]

\item[2. ]  For the form of the dual problem that follows the definition of semidefinite program (Dual Problem 2), we have that there is a solution that strictly satisfies the constraint, which can be obtained again by letting each of the diagonal blocks of our solution be an appropriate multiple of the identity. For example, we can have:

\begin{align*}
 Y_r &  =  \left(\norm{\sum_i v_i P_i}+1 \right) \mathbb{I}_{\Y_{1 \ldots r-1} \otimes \X_{1 \ldots r}} \\
   Y_{r-1} & =  2\dim{\X_r} \left(\norm{\sum_i v_i P_i}+1\right)\mathbb{I}_{\Y_{1 \ldots r-2} \otimes \X_{1 \ldots r-1}} \\
   Y_{r-2} & =  4\dim{\X_{r-1}}\dim{\X_r}\left(\norm{\sum_i v_i P_i}+1\right)\mathbb{I}_{\Y_{1 \ldots r-3} \otimes \X_{1 \ldots r-2}}   \\
   \vdots \\
  Y & =  2^{r-1}\dim{\X_2} \ldots \dim{\X_{r-1}}\dim{\X_r}\left(\norm{\sum_i v_i P_i}+1\right)\mathbb{I}_{ \X_{1}}  
\end{align*}

\end{mylist}

\noindent We will then write $v$ to refer to the optimal value of these optimizations problems.

We consider now the situation in which $n$ copies of the same interaction occur, with Alice acting independently, while Bob is free to correlate his actions in the different repetitions. Bob is trying to maximize the expected value per repetition that he achieves, with the value of a series of outcomes being the sum of their individual values. Using the characterization of the possible processes by which Bob can operate in a parallel repetition situation from Lemma \ref{cr:par}, we obtain that this corresponds to the optimization problem:

\begin{center}
  \begin{minipage}{2.6in}
    \centerline{\underline{Primal Problem 3}}\vspace{-7mm}
    \begin{align*}
      \text{maximize:} & \sum_{i_1, i_2, \ldots, i_n} \frac{1}{n} \left( v_{i_1} + v_{i_2} + \ldots + v_{i_n} \right) \ip{P_{i_1} \otimes P_{i_2}  \otimes \ldots \otimes P_{i_n} }{X}\\
      \text{subject to:}    \quad & \tr_{\Y_1^{\otimes n}}(X_1) = \mathbb{I}_{\X_1^{\otimes n}}\\
     \quad &  \tr_{\Y_2^{\otimes n}}(X_2) -X_1 \otimes \mathbb{I}_{\X_2^{\otimes n}} = 0\\
    \quad &  \vdots \\
      \quad &\tr_{\Y_r^{\otimes n}}(X) -X_{r-1} \otimes \mathbb{I}_{\X_r^{\otimes n}}  = 0\\    
	& X\in\pos{ \Y_{1 \ldots r}^{\otimes n} \otimes \X_{1 \ldots r}^{\otimes n} }
	, X_i \in\pos{\Y_{1 \ldots i}^{\otimes n} \otimes \X_{1 \ldots i}^{\otimes n} }
    \end{align*}
  \end{minipage}
\end{center}

Going through the same process to obtain a simplified version of the dual as in the case with a single repetition, we have that the dual of this optimization problem is:

\begin{center}
  \begin{minipage}{2.6in}
    \centerline{\underline{Dual Problem 3}}\vspace{-7mm}
    \begin{align*}
      \text{minimize:} & \tr(Y)  \\
      \text{subject to:}  \quad &  Y  \otimes \mathbb{I}_{\Y_1^{\otimes n}} - \tr_{\X_2^{\otimes n}}(Y_2) \geq 0 \\
      \quad &   Y_2  \otimes \mathbb{I}_{\Y_2^{\otimes n}} - \tr_{\X_3^{\otimes n}}(Y_3)  \geq 0 \\
      \vdots  \\
       \quad &   Y_r \otimes \mathbb{I}_{\Y_r^{\otimes n}} \geq    \sum_{i_1, i_2, \ldots, i_n} \frac{1}{n} \left( v_{i_1} + v_{i_2} + \ldots + v_{i_n} \right) \left( P_{i_1} \otimes P_{i_2}  \otimes \ldots \otimes P_{i_n}   \right) \\ 
       	& Y \in\herm{\X_1^{\otimes n} }
	, Y_i \in\herm{\Y_{1 \ldots i-1}^{\otimes n} \otimes \X_{1 \ldots i}^{\otimes n} }
    \end{align*}
  \end{minipage}
\end{center}

We again have that strong duality holds, and there are optimal solutions for both the primal and dual problems. To see this, we can adapt the forms of the primal and dual problems in the same way that we did for the case with one single repetition, and obtain Slater points by letting the diagonal blocks of our solutions be multiples of the identity. We write $v'$ then to refer to the optimal value of these optimizations problems.

\noindent We can now formally phrase the question

\begin{mylist}{\parindent}
\item[\ ] Can Bob improve on his  expected value per interaction when $n$ interactions are played in parallel, as opposed to a single interaction?
\end{mylist}

\noindent as

\begin{mylist}{\parindent}
\item[\ ] Is $v = v'$?
\end{mylist}

\noindent We will see now that the answer to this question is affirmative. Informally, it is clear that $v' \geq v$, since if Bob just plays his optimal strategy for one repetition in an independent way, his expected value per repetition will be the optimum expected value when only one single repetition occurs. And indeed, let $X, \{X_i\}$ represent an optimal solution to the primal version of the optimization problem for a single repetition (Primal Problem 1). Then, I claim that $X^{\otimes n}, \{X_i^{\otimes n}\}$ represent a feasible solution to the primal optimization problem for more than one repetition (Primal Problem 3), with value $v$. Indeed, using the properties of the tensor product that we stated in Chapter \ref{ch:back}, we have:

\begin{mylist}{\parindent}
\item[$\bullet$]  If we consider a random variable $V$ that takes value $v_i$ with probability $\ip{P_i}{X} $, we have that  

\begin{eqnarray*}
 & & \sum_{i_1, i_2, \ldots, i_n} \frac{1}{n} \left( v_{i_1} + v_{i_2} + \ldots + v_{i_n} \right) \ip{P_{i_1} \otimes P_{i_2}  \otimes \ldots \otimes P_{i_n} }{X^{\otimes n}}  \\
& = & \frac{1}{n}  \sum_{i_1, i_2, \ldots, i_n}  \left( v_{i_1} + v_{i_2} + \ldots + v_{i_n} \right) \ip{P_{i_1}}{X} \ip{P_{i_2}}{X} \ldots \ip{P_{i_n}}{X}   \\
& =  &  \frac{1}{n} \text{E}[nV] = \frac{1}{n} n \text{E}[V] = \text{E}[V] = v
\end{eqnarray*}

\item[$\bullet$]  $\tr_{\Y_1^{\otimes n}}(X_1^{\otimes n})= \left( \tr_{\Y_1}(X_1) \right) ^{\otimes n} = \mathbb{I}_{\X_1}^{\otimes n} =  \mathbb{I}_{\X_1^{\otimes n}}$
\item[\ ]  $\tr_{\Y_2^{\otimes n}}(X_2^{\otimes n})  - X_1^{\otimes n} \otimes \mathbb{I}_{\X_2^{\otimes n}}=  \left( \tr_{\Y_1}(X_2) \right) ^{\otimes n} - \left( X_1 \otimes \mathbb{I}_{\X_2} \right) ^{\otimes n} $
\item[\ ] \hspace{4.8cm} =  $\left(  X_1 \otimes \mathbb{I}_{\X_2}  \right) ^{\otimes n} - \left( X_1 \otimes \mathbb{I}_{\X_2} \right) ^{\otimes n}  = 0$
\item[\ ] \hspace{0.5cm} \vdots
\item[\ ]  $\tr_{\Y_r^{\otimes n}}(X^{\otimes n})  - X_{r-1}^{\otimes n} \otimes \mathbb{I}_{\X_r^{\otimes n}}= \tr_{\Y_r}(X)^{\otimes n} - \left( X_{r-1} \otimes \mathbb{I}_{\X_r} \right) ^{\otimes n} $
\item[\ ] \hspace{4.9cm} =  $\left(  X_{r-1} \otimes \mathbb{I}_{\X_r}  \right) ^{\otimes n} - \left( X_{r-1} \otimes \mathbb{I}_{\X_r} \right) ^{\otimes n}  = 0$
\item[$\bullet$ ] As $X \geq 0$ and $X_i \geq 0$, $X^{\otimes n} \geq 0$ and $X_i ^{\otimes n} \geq 0$. 

\end{mylist}

It is harder to come up with an intuitive reason for why $v' \leq v$. A possible informal argument would be that if Bob can obtain a value better than $v$ as the expected value per repetition, then in one of his repetitions his expected value is better than $v$. Then, to obtain a expected value better than $v$ when only one repetition is considered, he could simulate a setting in which multiple repetitions are considered, and let the "real" repetition of those be the one in which he obtains a value better than $v$. To give a formal argument for the fact that $v' \leq v$, we simply derive a solution to the dual problem for multiple repetitions (Dual Problem 3) with value $v$. Let then $Y$, $\{Y_i\}$ represent an optimal solution to Dual Problem 1. Then, I claim that 

\[
 \frac{1}{n} \left( Y \otimes \rho \otimes \ldots \otimes \rho + \ldots + \rho \otimes \rho \otimes \ldots \otimes Y  \right), \left\{  \frac{1}{n} \left( Y_i \otimes R_i \otimes R_i \otimes R_i + \ldots +  R_i \otimes R_i \otimes \ldots \otimes Y_i  \right) \right\}
\]

 \noindent represents a feasible solution to the dual optimization problem for more than one repetition (Dual Problem 3) with value $v$.  Indeed, using again the properties of the tensor product that we stated in Chapter \ref{ch:back}, we have:

\begin{mylist}{\parindent}
\item[$\bullet$]  $\tr \left( \frac{1}{n} \left( Y \otimes \rho \otimes \ldots \otimes \rho + \ldots + \rho \otimes \rho \otimes \ldots \otimes Y  \right) \right) = \frac{1}{n} n \tr(Y) \tr(\rho)^{n-1} = \tr(Y)=v$
\item[$\bullet$]  For the first constraint of the dual problem, we have:
\begin{eqnarray*}
 & &  \hspace{-1cm}   \frac{1}{n} \left( Y \otimes \rho \otimes \ldots \otimes \rho + \ldots + \rho \otimes \rho \otimes \ldots \otimes Y  \right) \otimes \mathbb{I}_{\Y_1^{\otimes n}}  -  \\
& &   \hspace{-1cm} \tr_{\X_2^{\otimes n}}  \left( \frac{1}{n} \left( Y_2 \otimes R_2 \otimes \ldots \otimes R_2 + \ldots + R_2 \otimes R_2 \otimes \ldots \otimes Y_2  \right) \right) \\
& &   \hspace{-1cm} =  \frac{1}{n} \left( \left( Y \otimes \mathbb{I}_{\Y_1} \right)  \otimes  \left( \rho \otimes \mathbb{I}_{\Y_1} \right) \otimes \ldots \otimes   \left( \rho \otimes \mathbb{I}_{\Y_1} \right) - \tr_{\X_2}(Y_2) \otimes  \tr_{\X_2}(R_2)  \otimes \ldots \otimes  \tr_{\X_2}(R_2) \right) + \ldots  \\
& &  \hspace{-1cm} +  \frac{1}{n} \left( \left( \rho \otimes \mathbb{I}_{\Y_1} \right)  \otimes  \ldots \otimes \left( \rho \otimes \mathbb{I}_{\Y_1} \right) \otimes   \left( Y  \otimes \mathbb{I}_{\Y_1} \right) - \tr_{\X_2}(R_2)  \otimes \ldots \otimes  \tr_{\X_2}(R_2)  \otimes  \tr_{\X_2}(Y_2) \right)    \\
& &   \hspace{-1cm} \geq 0, \text{as }   \rho \otimes \mathbb{I}_{\Y_1} = \tr_{\X_2}(R_2) \text{ and }  Y \otimes \mathbb{I}_{\Y_1} \geq \tr_{\X_2}(Y_2)
\end{eqnarray*}

\item[$\bullet$ ]  For successive constraints, we have:
\begin{eqnarray*}
 & &   \hspace{-1cm}  \frac{1}{n} \left( Y_i \otimes R_i \otimes \ldots \otimes R_i+ \ldots + R_i \otimes R_i \otimes \ldots \otimes Y  \right) \otimes \mathbb{I}_{\Y_i^{\otimes n}}  -  \\
& &  \hspace{-1cm} \tr_{\X_{i+1}^{\otimes n}}  \left( \frac{1}{n} \left( Y_{i+1} \otimes R_{i+1} \otimes \ldots \otimes R_{i+1} + \ldots + R_{i+1} \otimes R_{i+1} \otimes \ldots \otimes Y_{i+1}  \right) \right) \\
&    \hspace{-2cm}= & \hspace{-1cm}   \frac{1}{n}  \left( Y_i \otimes \mathbb{I}_{\Y_i} \right)  \otimes  \left( R_i \otimes \mathbb{I}_{\Y_i} \right) \otimes \ldots \otimes   \left( R_i \otimes \mathbb{I}_{\Y_i} \right) \\ 
& & \hspace{-1cm} - \frac{1}{n} \tr_{\X_{i+1}}(Y_{i+1}) \otimes  \tr_{\X_{i+1}}(R_{i+1})  \otimes \ldots \otimes  \tr_{\X_{i+1}}(R_{i+1}) + \ldots  \\
& &  \hspace{-1cm}  +  \frac{1}{n}  \left( R_i \otimes \mathbb{I}_{\Y_i} \right)  \otimes  \ldots \otimes \left( R_i \otimes \mathbb{I}_{\Y_i} \right) \otimes   \left( Y_i  \otimes \mathbb{I}_{\Y_i} \right) \\
& & \hspace{-1cm} - \frac{1}{n} \tr_{\X_{i+1}}(R_{i+1})  \otimes \ldots \otimes  \tr_{\X_{i+1}}(R_{i+1})  \otimes  \tr_{\X_{i+1}}(Y_{i+1})     \\
&  \hspace{-2cm} \geq &  \hspace{-1cm} 0, \text{as }   R_i \otimes \mathbb{I}_{\Y_i} = \tr_{\X_{i+1}}(R_{i+1}) \text{ and }  Y_i \otimes \mathbb{I}_{\Y_i} \geq \tr_{\X_{i+1}}(Y_{i+1})
\end{eqnarray*}
\item[$\bullet$ ] For the last constraint,  As $ Y_r \otimes \mathbb{I}_{\Y_r}  \geq \sum_i v_i P_i$, and $R_r \otimes \mathbb{I}_{\Y_r} = \sum_i P^i $  we have that 
\[
 \left( R_r \otimes \mathbb{I}_{\Y_r} \right)^{\otimes k-1} \otimes  \left( Y_r \otimes \mathbb{I}_{\Y_r}  \right) \otimes  \left( R_r \otimes \mathbb{I}_{\Y_r} \right)^{\otimes n-k} \geq  \sum_{i_1, i_2, \ldots, i_n} v_{i_k}  \left( P_{i_1} \otimes P_{i_2}  \otimes \ldots \otimes P_{i_n}   \right)
 \]
 
Therefore,

\begin{eqnarray*}
& & \frac{1}{n} \sum_k \left( R_r \otimes \mathbb{I}_{\Y_r} \right)^{\otimes k-1} \otimes  \left( Y_r \otimes \mathbb{I}_{\Y_r}  \right) \otimes  \left( R_r \otimes \mathbb{I}_{\Y_r} \right)^{\otimes n-k} \\
& \geq & \sum_{i_1, i_2, \ldots, i_n} \frac{1}{n} \left( v_{i_1} + \ldots + v_{i_n} \right) \left( P_{i_1} \otimes P_{i_2}  \otimes \ldots \otimes P_{i_n}   \right)
 \end{eqnarray*}

\end{mylist}

Note that this implies that the answer to our question when only classical information is allowed is positive as well, as the classical case is a particular case of the quantum one.

%% file: 05-notindependent.tex
\newpage

\chapter{Optimal strategies in risk-minimizing parallel repetition}\label{ch:nindep}

We consider here the situation in which the outcomes are split into two groups,  \emph{winning} outcomes and \emph{losing} outcomes, and Bob desires to obtain winning outcomes 

When Bob is trying to optimize the expected number of repetitions in which he obtains a winning outcome, the best he can do is to play independently several copies of his optimal strategy for achieving a winning outcome when only one repetition is considered. This can be seen from assigning value $1$ to the winning outcome and value $0$ to all other outcomes, and considering our result in the previous chapter. However, we can also consider the case in which Bob is not concerned with optimizing the number of repetitions in which he obtains the winning outcome, but rather with making sure that the number of repetitions in which he obtains the winning outcome is above a certain threshold. We can ask whether it is still optimal for Bob to play independently in this case.

To answer this question, note that we can assume without loss of generality that there are only two outcomes, by grouping together all the outcomes that correspond to a winning situation, and grouping also together all outcomes that correspond to a losing situation. To express this situation formally, consider a specific description of a game in the way presented in Chapter \ref{ch:sdp}, with $P_0$ and $P_1$ being the operators from Lemma \ref{lm:rep} that corresponds to the losing and winning outcome, respectively. We have than that determining the optimal process for Bob when he is trying to maximize the probability that he obtains the winning outcome, and only one repetition of the interaction is considered, corresponds to the following optimization problem:

\begin{center}
  \begin{minipage}{2.6in}
    \centerline{\underline{Primal Problem 4}}\vspace{-7mm}
    \begin{align*}
      \text{maximize:} & \ip{P_1}{X}\\
      \text{subject to:}    \quad & \tr_{\Y_1}(X_1) = \mathbb{I}_{\X_1}\\
     \quad &  \tr_{\Y_2}(X_2) -X_1 \otimes \mathbb{I}_{\X_2} = 0\\
    \quad &  \vdots \\
      \quad &\tr_{\Y_r}(X) -X_{r-1} \otimes \mathbb{I}_{\X_r}  = 0\\    
	& X\in\pos{\Y_{1 \ldots r} \otimes \X_{1 \ldots r} }
	, X_i \in\pos{\Y_{1 \ldots i} \otimes \X_{1 \ldots i} }
    \end{align*}
  \end{minipage}
\end{center}

The process of computing the dual is identical to the one one in Chapter \ref{ch:avg}, and we obtain as a result the dual problem:

\begin{center}
  \begin{minipage}{2.6in}
    \centerline{\underline{Dual Problem 4}}\vspace{-7mm}
    \begin{align*}
      \text{minimize:} & \tr(Y)  \\
      \text{subject to:}  \quad &  Y  \otimes \mathbb{I}_{\Y_1} - \tr_{\X_2}(Y_2) \geq 0 \\
      \quad &   Y_2  \otimes \mathbb{I}_{\Y_2} - \tr_{\X_3}(Y_3)  \geq 0 \\
      \vdots  \\
       \quad &   Y_r \otimes \mathbb{I}_{\Y_r} \geq  P_1 \\ 
       	& Y \in\herm{\X_1 }
	, Y_i \in\herm{\Y_{1 \ldots i-1} \otimes \X_{1 \ldots i} }
    \end{align*}
  \end{minipage}
\end{center}

We have again strong duality, with optimal solutions existing for both the primal and the dual problem. To see this, we can just notice that these programs are a particular case of the ones for the situation in Chapter \ref{ch:avg}. We will then denote by $p$ the optimal value of these semidefinite programs.

We consider now the situation in which several independent copies of the same interaction occur in parallel, and Bob is trying to optimize his probability of obtaining a winning outcome in at least $k$ of then. In our analysis of the situation, we will use $\Sigma^n_k$ to denote the subset of $\{0,1\}^n$ corresponding to the elements with exactly $k$ $1$s, and  $\Sigma^n_{\geq k}$ to denote the subset of $\{0,1\}^n$  corresponding to the elements with at least $k$ $1$s. From Lemma \ref{cr:par}, we obtain then that this situation corresponds to the following optimization problem:

\begin{center}
  \begin{minipage}{2.6in}
    \centerline{\underline{Primal Problem 5}}\vspace{-7mm}
    \begin{align*}
      \text{maximize:} & \sum_{ \left( i_1, i_2, \ldots, i_n  \right) \in \Sigma^n_{\geq k}} \ip{P_{i_1} \otimes P_{i_2}  \otimes \ldots \otimes P_{i_n} }{X}\\
      \text{subject to:}    \quad & \tr_{\Y_1^{\otimes n}}(X_1) = \mathbb{I}_{\X_1^{\otimes n}}\\
     \quad &  \tr_{\Y_2^{\otimes n}}(X_2) -X_1 \otimes \mathbb{I}_{\X_2^{\otimes n}} = 0\\
    \quad &  \vdots \\
      \quad &\tr_{\Y_r^{\otimes n}}(X) -X_{r-1} \otimes \mathbb{I}_{\X_r^{\otimes n}}  = 0\\    
	& X\in\pos{ \Y_{1 \ldots r}^{\otimes n} \otimes \X_{1 \ldots r}^{\otimes n} }
	, X_i \in\pos{\Y_{1 \ldots i}^{\otimes n} \otimes \X_{1 \ldots i}^{\otimes n} }
    \end{align*}
  \end{minipage}
\end{center}

The process to obtain a simplified version of the dual as in the case with a single repetition gives us now:

\begin{center}
  \begin{minipage}{2.6in}
    \centerline{\underline{Dual Problem 5}}\vspace{-7mm}
    \begin{align*}
      \text{minimize:} & \tr(Y)  \\
      \text{subject to:}  \quad &  Y  \otimes \mathbb{I}_{\Y_1^{\otimes n}} - \tr_{\X_2^{\otimes n}}(Y_2) \geq 0 \\
      \quad &   Y_2  \otimes \mathbb{I}_{\Y_2^{\otimes n}} - \tr_{\X_3^{\otimes n}}(Y_3)  \geq 0 \\
      \vdots  \\
       \quad &   Y_r \otimes \mathbb{I}_{\Y_r^{\otimes n}} \geq    \sum_{ \left( i_1, i_2, \ldots,i_n \right)  \in \Sigma^n_{\geq k}}   P_{i_1} \otimes P_{i_2}  \otimes \ldots \otimes P_{i_n}  \\ 
       	& Y \in\herm{\X_1^{\otimes n} }
	, Y_i \in\herm{\Y_{1 \ldots i-1}^{\otimes n} \otimes \X_{1 \ldots i}^{\otimes n} }
    \end{align*}
  \end{minipage}
\end{center}

We have again strong duality with optimal solutions being achieved, as can be seen in the same way as for the problems in Chapter \ref{ch:avg}, that is, making each of the elements of our solution be a block of a larger matrix so that we have programs in the form in which we state the theorem for the existence of Slater points, and then letting the Slater points correspond to multiples of the identity. We will then denote by $p'$ the optimal value of these problems.

Intuitively, the value of $p'$ will be at least  $\sum_{ k \leq t \leq n} {{n}\choose{t}} p^t (1-p)^{n-t}$, since that is what a process for Bob that repeats $n$ independent copies of the optimal process for one interaction would achieve. And indeed, if we let a solution to Primal Problem 4 be given by $X,\{X_i\}$, then $X^{\otimes n},\{X_i^{\otimes n}\}$ gives us a solution to Primal Problem 5 with value $\sum_{ k \leq t \leq n} {{n}\choose{t}} p^t (1-p)^{n-t}$. We can then formally phrase the question:

\begin{mylist}{\parindent}
\item[\ ] Is it optimal for Bob to play independently when trying to force a certain outcome in at least $k$ out of $n$ independent parallel copies of an interaction?
\end{mylist}

as

\begin{mylist}{\parindent}
\item[\ ] Is $p' = \sum_{ k \leq t \leq n} {{n}\choose{t}} p^t (1-p)^{n-t}$?
\end{mylist}

This has been established in the literature to be the case when $k=n$  \cite{Gutoski10,MittalS07}. The way in which this is done is by letting an optimal solution to Dual Problem 4 be given by $Y,\{Y_i\}$, and then considering the solution $Y^{\otimes n},\{Y_i^{\otimes n}\}$ to Dual Problem 5. As the right hand side of the last constraint is $P_1^{\otimes n} \geq \mathbb{I}_{\Y_r} \otimes Y \geq 0$, the properties of the tensor product that we mention in Chapter \ref{ch:back} are enough to determine that $Y^{\otimes n},\{Y_i^{\otimes n}\}$ is indeed a feasible solution, with value $p^n$. A natural way to extend this to the case in which $k<n$ would be to let our solution be 

\[
\sum_{\left( i_1, \ldots, i_n \right) \in \Sigma_{\geq k}^n} f(i_1) \otimes \ldots \otimes f(i_n), \left\{ \sum_{\left( i_1, \ldots, i_n \right) \in \Sigma_{\geq k}^n}  f_i(i_1) \otimes \ldots \otimes f_i(i_n) \right\},
\]

\noindent where $f(0)=\rho-Y$, $f(1)=Y$, $f_i(0)=R_i-Y_i$ and $f_i(1)=Y_i$. It is not clear that this would be a feasible solution. However, if we make the assumption that all the constraints except the last one  are satisfied with equality in the Dual Problem 4 for $Y^{\otimes n},\{Y_i^{\otimes n}\}$, which from Lemma 3.13 in \cite{Gutoski10} is a valid assumption to make, we will have that  all the constraints except the last one of Dual Problem 5 are satisfied by this solution. However, it is still not clear how to prove that the proposed solution does actually satisfy the last constraint. If at this point we could make the additional assumption that $Y \leq \rho$ and $Y_i \leq R_i$, the proof that we will give later for the classical case would give us that the candidate we are considering is indeed a feasible solution. However, there are cases in which it is not possible to make this assumption, as it follows from the existence of the counterexample to $p' = \sum_{ k \leq t \leq n} {{n}\choose{t}} p^t (1-p)^{n-t}$ that we show now.

\section{Counterexample to the independence of acting optimally for Bob}

It is indeed possible to find a simple example, with $n=2$, $r=1$ and $k=1$, in which the value of $p$ is $\cos^2(\pi/8) \approx 0.85$, but the optimal probability for Bob to obtain a winning outcome in one of at least two repetitions of the interaction is not $\cos^2(\pi/8) + 2\cos^2(\pi/8) \sin^2(\pi/8) \approx 0.98$. Instead, it is equal to $1$. A single repetition of the interaction corresponding to this example follows the following process:

\begin{mylist}{\parindent}
\item[1.]
Alice prepares a pair of qubits $(\reg{X},\reg{Z})$ in the state
\[
u = \frac{1}{\sqrt{2}} \ket{00} + \frac{1}{\sqrt{2}} \ket{11}
\in\X\otimes\Z,
\]

\noindent and sends $\reg{X}$ to Bob.

\item[2.]
Bob applies a quantum channel of his choice to $\reg{X}$, obtaining a
qubit $\reg{Y}$ that he sends back to Alice.
After this action, the pair $(\reg{Y},\reg{Z})$ will be in some
particular state $\sigma\in\density{\Y\otimes\Z}$.

\item[3.]
Alice measures $(\reg{Y},\reg{Z})$ with respect to the projective
measurement $\{\Pi_0,\Pi_1\}$, where $\Pi_0$ corresponds to the losing outcome, while $\Pi_1$ corresponds to the winning outcome. $\Pi_0 = \mathbb{I} - \Pi_1$ and
$\Pi_1 = v v^{\ast}$, for
\[
v = \cos(\pi/8) \ket{00} + \sin(\pi/8) \ket{11}.
\]

\end{mylist}

The probability that Bob obtains the winning outcome is
\[
\ip{\Pi_1}{\sigma} = \fid(v v^{\ast},\sigma)^2,
\]
where $\fid(\cdot,\cdot)$ denotes the \emph{fidelity} function $\fid(P,Q) =  \norm{\sqrt{P}\sqrt{Q}}_1$ and
we have the equality from the fact that $v v^{\ast}$ is pure.

Now, if Bob makes
$\sigma\in\density{\Y\otimes\Z}$ be the state after step 2, it must hold that

\[ 
\tr_{\Y}(\sigma) =  \tr_{\X}(u u^{\ast}) = \frac{1}{2} \mathbb{I}_{\Z}.
\]

It is known that the fidelity function is monotone under partial tracing, so we have then that
\[
\fid(v v^{\ast},\sigma)^2 \leq \fid\left( \tr_{\Y}(v v^{\ast}),
\tr_{\Y}(\sigma)\right)^2
= \fid(Q,R)^2
\]
for
\[
Q = 
\begin{pmatrix}
  \cos^2(\pi/8) & 0 \\ 0 & \sin^2(\pi/8)
\end{pmatrix}
\quad\quad\text{and}\quad\quad
R = 
\begin{pmatrix}
  \frac{1}{2} & 0 \\ 0 & \frac{1}{2}
\end{pmatrix}.
\]

Computing $\sqrt{Q}\sqrt{R}$, we have then that

\begin{eqnarray*}
\fid(Q,R)^2 &  = & \norm{\sqrt{Q}\sqrt{R}}_1^2  = \frac{1}{2}\left(\cos(\pi/8) + \sin(\pi/8)\right)^2 \\
&   = & \frac{1 + \sin(\pi/4) }{2} = \frac{1 + \cos(\pi/4) }{2}  =  \cos^2(\pi/8),
\end{eqnarray*}

\noindent using the trigonometrical identity $\cos(\alpha/2) = \sqrt{ \frac{1 + \cos \alpha}{2}}$ in the last equality, and the identity  $\sin(2\alpha) = 2 \sin(\alpha)\cos(\alpha)$ in the third equality. We have then that the optimal probability for Bob of achieving the winning outcome is at most $\cos^2(\pi/8) \approx 0.85$. This bound is actually tight, since if Bob acts as the identity, he achieves the winning outcome with probability

\[
\ip{vv^*}{uu^*}^2 = \frac{ \left(\cos(\pi/8)+\sin(\pi/8) \right)^2}{2} = \cos^2(\pi/8).
\]

Now, for two instantiations of the interaction described above in which Alice operates independently, we
consider what happens when Bob applies the phase flip
$\ket{00} \mapsto -\ket{00}$,
$\ket{01} \mapsto \ket{01}$,
$\ket{10} \mapsto \ket{10}$,
$\ket{11} \mapsto \ket{11}$
on the two qubits he receives. The state he receives is
\[
\frac{1}{2}\ket{0000} +
\frac{1}{2}\ket{0011} +
\frac{1}{2}\ket{1100} +
\frac{1}{2}\ket{1111}
\]
and Bob's phase flip transforms this state to
\[
-\frac{1}{2}\ket{0000} +
\frac{1}{2}\ket{0011} +
\frac{1}{2}\ket{1100} +
\frac{1}{2}\ket{1111}.
\]
Writing
\[
w = - \sin(\pi/8)\ket{00} + \cos(\pi/8)\ket{11}
\]
we find that
\[
-\frac{1}{2}\ket{0000} +
\frac{1}{2}\ket{0011} +
\frac{1}{2}\ket{1100} +
\frac{1}{2}\ket{1111}
=
\frac{1}{\sqrt{2}}v\otimes w
+
\frac{1}{\sqrt{2}}w\otimes v.
\]
When Alice measures this state with respect to the measurement
$\{\Pi_0,\Pi_1\}$, there will then be exactly one winning outcome and one failing outcome. Bob passes  (and fails) exactly one of the two tests with certainty. The ability of Bob to correlate his answers in this way is suggestive of a perfect form of hedging, where the risk of a loss in one game of chance is perfectly offset the actions in a second game. 

\noindent Note that any strategy for Bob in which he does better that playing independently when trying to win at least $k$ times out of $n$ will imply the existence of a hedging phenomenon, in the sense that for this strategy  there will be a $k'$ for which Bob will do worse than when playing independently at winning at least $k'$ times out of $n$. This follows from our result of Chapter \ref{ch:avg} that it is optimal for Bob to play independently if he is trying to maximize his expected number of wins.

\section{Analysis in the classical case}

In the classical case, it is not possible to find an example like the one we just presented. This can be derived from our analysis of the situation using semidefinite programs, observing first that in the classical case there exists an optimal dual solution  $Y,\{Y_i\}$  to Dual Problem 4 in which all of the blocks in our solution are diagonal matrices. To see this, consider an arbitrary solution $Y',\{Y'_i\}$ to Dual Problem 4. Now, I claim that   $\Lambda(Y'),\{\Lambda(Y'_i)\}$ is a feasible solution with the same value, where the dephashing channel $\Lambda$ sets the non-diagonal entries of the input to zero, and leaves the diagonal entries unaltered, giving then as the output an operator represented by a diagonal matrix. Indeed, we have that 

\begin{mylist}{\parindent}
\item[$\bullet$ ] $\tr(\Lambda(Y)) = \tr(Y)$, since the diagonal elements of the corresponding matrices are the same.
\item[$\bullet$ ] $\Lambda$ is a positive operator, since the diagonal elements of a positive semidefinite matrix are non-negative. Then, as $\Lambda$ commutes with the partial trace, $\Lambda(P_1)=P_1$, and $\Lambda(\mathbb{I} \otimes A) = \mathbb{I} \otimes \Lambda(A)$ for any operator $A$, we have that all the constraints are satisfied. This is because then we can write them as $\Lambda(A-B) \geq 0$, with $A \geq B$, and the constraint being therefore satisfied.
\end{mylist}

Now that we make the assumption that $Y,\{Y_i\}$ correspond to diagonal matrices  (remember that $P_0$, $P_1$, $\{R_i\}$ and $\rho$ do as well), we have that we can make the additional assumption that $Y \leq \rho,  Y_i \leq R_i $. Indeed, consider any solution to Dual Problem 4 with operators  $Y,\{Y_i\}$ that correspond to diagonal matrices. Then, I claim that if we let $Y'$ be the element-wise minimum of $Y$ and $\rho$, and $Y_i'$ be the element-wise minimum of $Y_i$ and $R_i$ (note that then $Y' \leq \rho$ and $Y_i' \leq R_i$) , $Y',\{Y_i'\}$ is a feasible solution  to Dual Problem 4 with a value equal to at most the one of $Y,\{Y_i\}$. Indeed, we have

\begin{mylist}{\parindent}
\item[$\bullet$ ] $Y' \leq Y$, so $\tr(Y') \leq \tr(Y)$.
\item[$\bullet$ ]  An element along the diagonal of the matrix for $Y'  \otimes \mathbb{I}_{\Y_1} $ is equal either to the element in the same position for  $Y \otimes \mathbb{I}_{\Y_1}$, or to the element in the same position for $\rho \otimes \mathbb{I}_{\Y_1}$.

In case it is equal to the corresponding element of $Y \otimes \mathbb{I}_{\Y_1}$, from the feasibility of $Y, \{Y_i\}$ we have that this is at least the corresponding element of $ \tr_{\X_{2}} (Y_2)$, which will be at least the corresponding element  of $ \tr_{\X_2} (Y'_2)$, since $Y_2 \geq Y'_2$. We have then that in this case the element that we are considering of  the matrix for $Y'  \otimes \mathbb{I}_{\Y_1} $ will be at least equal to the element in the same position for the matrix for  $ \tr_{\X_{2}} (Y'_2)$

In case it is equal to the corresponding element of $\rho \otimes \mathbb{I}_{\Y_1}$, we have that it then equal to the element in the same position for $\tr_{\X_2}(R_2)$. This is at least equal to the corresponding element for $\tr_{\X_2}(Y'_2)$, from the definition of $Y'_2$. We have then again that the element that we are considering of  the matrix for $Y'  \otimes \mathbb{I}_{\Y_1} $ will be at least equal to the element in the same position for the matrix for  $ \tr_{\X_{2}} (Y'_2)$.

As we are dealing with diagonal matrices, this establishes that the first constraint is satisfied. A similar argument (replacing $Y$ by $Y_i$ and $\rho$ by $R_i$) gives us that all other constraints except the last one are satisfied.

\item[$\bullet$ ]  An element along the diagonal of the matrix for $Y_r'  \otimes \mathbb{I}_{\Y_r} $ is equal either to the element in the same position for  $Y_r \otimes \mathbb{I}_{\Y_r}$, or to the element in the same position for $R_r \otimes \mathbb{I}_{\Y_r}$.

In case it is equal to the corresponding element of $Y_r \otimes \mathbb{I}_{\Y_r}$, from the feasibility of $Y, \{Y_i\}$ we have that this is at least the corresponding element of $P_1$. 

In case it is equal to the corresponding element for $R_r \otimes \mathbb{I}_{\Y_r}$, we have that as $P_1 \leq R_r \otimes \mathbb{I}_{\Y_r}$, the element we are considering of  the matrix for $Y_r  \otimes \mathbb{I}_{\Y_r} $ is at least equal to the corresponding element for $P_1$.

\end{mylist}

\noindent Making then all these assumptions about $Y, \{Y_i\}$, we have that we can prove that 

\[
\sum_{\left( i_1, \ldots, i_n \right) \in \Sigma_{\geq k}^n} f(i_1) \otimes \ldots \otimes f(i_n), \left\{ \sum_{\left( i_1, \ldots, i_n \right) \in \Sigma_{\geq k}^n}  f_i(i_1) \otimes \ldots \otimes f_i(i_n) \right\}
\]

\noindent is actually a feasible solution to Dual Problem 5, deriving then that

\[
p' = \sum_{ k \leq t \leq n} {{n}\choose{t}} p^t (1-p)^{n-t},
\]

\noindent as desired.  To prove that the solution is indeed feasible, we need the following Lemma:

\begin{lemma} \label{lm:pos}

Assume all of $A_0, A_1, R ,B_1=A_1+R, B_0=A_0-R$ are positive semidefinite operators.  Then for every choice of integers $n \geq 1$, $k \in \{0, \ldots, n\}$, it holds that

\[
 \sum_{ \left( i_1, \ldots, i_n \right) \in \Sigma_{\geq k}^n} B_{i_1} \otimes \ldots \otimes B_{i_n} \geq   \sum_{ \left( i_1, \ldots, i_n \right)  \in \Sigma_{\geq k}^n} A_{i_1} \otimes \ldots \otimes A_{i_n}
\]

\end{lemma}

\begin{proof}

 (by induction on $n$)

\begin{mylist}{\parindent}

 \item[$\bullet$ ] For $n=1$, we must consider the cases $k=0$ and $k=1$.

For $k=0$ we have

\[
B_0 + B_1 = (A_0-R) + (A_1+R) = A_0 + A_1
\]

and for $k=1$ we have
\[
B_1=A_1+R \geq A_1,
\]

as required.

\item[$\bullet$ ] For $n>1$, we have

\begin{eqnarray*}
  \sum_{ \left( i_1, \ldots, i_n \right) \in \Sigma_{\geq k}^n} B_{i_1} \otimes \ldots \otimes B_{i_n} &= &  \sum_{ \left( i_1, \ldots, i_{n-1} \right)\in \Sigma_{\geq k}^{n-1}}B_{i_1} \otimes \ldots \otimes B_{i_{n-1}} \otimes B_{0} \\
   &  & + \sum_{ \left( i_1, \ldots, i_{n-1} \right) \in \Sigma_{\geq k-1}^{n-1}}B_{i_1} \otimes \ldots \otimes B_{i_{n-1}} \otimes B_{1}
\end{eqnarray*}

Applying the induction hypothesis, we obtain

\begin{align*}
  \sum_{ \left( i_1, \ldots, i_n \right) \in \Sigma_{\geq k}^n} B_{i_1} \otimes \ldots \otimes B_{i_n}
  & \geq & &   \sum_{ \left( i_1, \ldots, i_{n-1} \right)\in \Sigma_{\geq k}^{n-1}} A_{i_1} \otimes \ldots \otimes A_{i_{n-1}} \otimes (A_0 - R) \\
 & & &  + \sum_{  \left( i_1, \ldots, i_{n-1} \right)\in \Sigma_{\geq k-1}^{n-1}}A_{i_1} \otimes \ldots \otimes A_{i_{n-1}} \otimes (A_1 + R) \\
   & = & & \sum_{ \left( i_1, \ldots, i_n \right) \in \Sigma_{\geq k}^n}  A_{i_1} \otimes \ldots \otimes A_{i_{n}}  \\
   & & & + \sum_{  \left( i_1, \ldots, i_{n-1} \right)\in \Sigma_{ k-1}^{n-1}} A_{i_1} \otimes \ldots \otimes A_{i_{n-1}} \otimes R  \\
    & \geq  & & \sum_{  \left( i_1, \ldots, i_n \right) \in \Sigma_{\geq k}^n}  A_{i_1} \otimes \ldots \otimes A_{i_{n}},  
\end{align*}

as required.
\end{mylist}

Note that if we substitute the set of binary strings with at least $k$ ones by any other monotone subset of $\{0,1\}^n$, the proof still holds.

\end{proof}

Using this Lemma, we can prove now the feasibility of the proposed solution to Dual Problem 5:

\begin{mylist}{\parindent}
\item[$\bullet$]  The fact that our proposed solution to Dual Problem 5 satisfies the first constraint follows from the Lemma, with $A_0= \tr_{\X_2} (R_2-Y_2)$, $A_1= \tr_{\X_2} (Y_2) $,  and $R= \mathbb{I}_{\Y_1} \otimes Y  - \tr_{\X_2}(Y_2) $, with $B_0$ being then $ \tr_{\X_2} (R_2) - \mathbb{I}_{\Y_1} \otimes Y = \mathbb{I}_{\Y_1} \otimes (\rho -Y)$, and $B_1$ being $ \mathbb{I}_{\Y_1} \otimes Y$. That $B_0 \geq 0$ follows from our assumption that $Y \leq \rho$. That $A_1$ is $\geq 0$ follows from the observation that we made before that all blocks of a feasible solution to the dual problems that we consider have to be positive semidefinite.

\item[$\bullet$] The fact that our proposed solution to Dual Problem 5 satisfies all constraints from the second to the second last one follows from the Lemma in the same way, with $A_0= \tr_{\X_{i+1}} (R_{i+1}-Y_{i+1})$, $A_1= \tr_{\X_{i+1}} (Y_{i+1}) $,  and $R= \mathbb{I}_{\Y_i} \otimes Y_i  - \tr_{\X_{i+1}}(Y_{i+1}) $. $B_0$ is now $\mathbb{I}_{\Y_i} \otimes (R_i -Y_i)$, and  $B_1$ is $ \mathbb{I}_{\Y_i} \otimes Y_i$. That $B_0 \geq 0$ follows from our assumption that $Y_i \leq R_i$. That $A_1$ is $\geq 0$ follows from the observation that we made before that all blocks of a feasible solution to the dual problems that we consider have to be positive semidefinite.

\item[$\bullet$] The fact that our proposed solution to Dual Problem 5 satisfies the last constraint follows from the Lemma as well, with $A_0=P_0$, $A_1=P_1$, and $R=\mathbb{I}_{\Y_r} \otimes Y_r - P_1$.  $B_0$ is then $ P_0  + P_1 - \mathbb{I}_{\Y_r} \otimes Y_r = \mathbb{I}_{\Y_1} \otimes (R_r -Y_r)$, and $B_1$ is $ \mathbb{I}_{\Y_r} \otimes Y_r $. That $B_0 \geq 0$ follows from our assumption that $Y_r \leq R_r$. 

\end{mylist}

We have then that our proposed solution to Dual Problem 5 is feasible in the classical case, as desired.

Note that to prove our bound  of $\sum_{t=k}^n {n \choose t} p^t (1-p)^{n-t}$, the only thing that we needed was a way to turn a feasible solution $Y,\{Y_i\}$ for Dual Problem 4 into another solution $Y', \{Y_i'\}$ with value $p' \leq p$ that satisfies $Y' \leq \rho$, $Y'_i \leq R_i$. We showed that this was possible in the case where $\{P_j\}, \{R_i\}, \rho, \{Y_i\}, Y$ are all diagonal matrices.  However, one could generalize our conditions further. If $Y$, $\{Y_i\}$ and $\rho$, $\{R_i\}$ were simultaneously diagonalizable, one could use a similar construction as before to define a new solution  $Y', \{Y_i'\}$, just by making them take the minimum of the eigenvalues of $Y$, $\{Y_i\}$ and $\rho$, $\{R_i\}$ in their shared eigenspaces.

Then, it is enough for our analysis of the first constraint to work out that after $Y_2$ is modified to obtain $Y_2'$,  $\tr_{\X_2}(Y_2')$ is diagonal in a basis where $\rho \otimes \mathbb{I}_{\Y_1} $ and $Y \otimes \mathbb{I}_{\Y_1} $ (and therefore $Y' \otimes \mathbb{I}_{\Y_1}$) are diagonal as well. Indeed, if that is the case, then we obtain feasibility for the first constraint of Dual Problem 5 from doing a change of basis to such a bassis where $\tr_{\X_2}(Y_2')$, $\rho \otimes \mathbb{I}_{\Y_r} $, $Y \otimes \mathbb{I}_{\Y_r} $  and $Y' \otimes \mathbb{I}_{\Y_r}$ are diagonal, and analyzing the two following cases:

\begin{mylist}{\parindent}

\item[$\bullet$]  An element along the diagonal of the matrix for $Y'  \otimes \mathbb{I}_{\Y_1} $ is equal to the element in the same position for  $Y \otimes \mathbb{I}_{\Y_1}$, which as we started with a valid solution, and we are talking about diagonal elements, is at least the corresponding element of $\tr_{\X_2}(Y_2)$. As $Y_2' \leq Y_2$ (so $\tr_{\X_2}(Y_2') \leq  \tr_{\X_2}(Y_2)$), and again, we are talking about diagonal elements, this is at least the corresponding diagonal element of  $\tr_{\X_2} (Y_2')$. 

\item[$\bullet$] An element along the diagonal of the matrix for $Y'  \otimes \mathbb{I}_{\Y_1} $ is equal to the element in the same position for  $\rho \otimes \mathbb{I}_{\Y_1} = \tr_{\X_2}(R_2) $.  Now,  $R_2 \geq Y_2'$ implies that $\tr_{\X_2}(Y'_2) \leq  \tr_{\X_2}(R_2)$. Then we have again that as we are talking about diagonal elements, the element of $Y'  \otimes \mathbb{I}_{\Y_1} $ we are considering is at least the corresponding element of $\tr_{\X_2}(Y'_2)$.

\end{mylist}

The same logic applies to the rest of the constraints, up to the last constraint, where we require that $P_1$ is diagonal in a basis where $R_{r-1}\otimes \mathbb{I}_{\Y_r}$ and $Y_{r-1} \otimes \mathbb{I}_{\Y_r}$ (and therefore $Y_{r-1}' \otimes \mathbb{I}_{\Y_r}$) are diagonal as well.

These more general conditions we are considering include as a subcase the case where there is a basis for $\X_1$ in which $\rho$ and $Y$ are diagonal, and by tensoring this base with a base for $\Y_1 \otimes \X_2$ we can obtain another one where $Y_2$ and $R_2$ are diagonal. Then, we can keep extending this basis in the same way in the next constraints, until we obtain a basis for $Y_{1 \ldots r -1} \otimes \X_{1 \ldots r}$ where $Y_r$ and $R_r$ are diagonal, and which can be extended to a basis where $P_1$ is diagonal. This subcase includes then the one where all the blocks $Y,\{Y_i\}$ of the dual solution are multiples of the identity, and so are the $\rho, \{R_i\}$. This occurs for example in the semidefinite programs studied in \cite{MolinaVW12}, where $r=1$, and the initial state sent from Alice to Bob is a superposition with equal probabilities of states taken from an orthogonal basis.

%% file: 06-upperbound.tex
\newpage

\chapter{Quantitative bounds to hedging phenomena}\label{ch:ub}

As we saw in Chapter \ref{ch:nindep}, the naive upper bound for the optimum probability for Bob of achieving the winning outcome in at least $k$ of  $n$ independent copies of an interaction  as a function of $p$, his optimum probability to win when only one copy is considered, does not actually hold. However, it is still possible to establish weaker bounds. We will do so in this chapter, proving first a bound of

\[
 \sum_{t=k}^n {n \choose t} p^t,
\]

\noindent and modifying then our method to prove a stronger bound of

\[
p^k {n \choose k}.
\]

The procedure that we will follow to obtain these bounds will be based on building a feasible solution to Dual Program 5 from an optimal solution to Dual Problem 4. Let then  $Y$, $\{Y_i\}$ represent an optimal solution to Dual Problem 4. I claim that it holds  that a feasible solution to Dual Problem 5 is given by 

\[
\sum_{\left( i_1, \ldots, i_n \right) \in \Sigma_{\geq k}^n} f(i_1) \otimes \ldots \otimes f(i_n), \left\{ \sum_{\left( i_1, \ldots, i_n \right) \in \Sigma_{\geq k}^n}  f_i(i_1) \otimes \ldots \otimes f_i(i_n) \right\},
\]

\noindent where $f(0)=\rho$, $f(1)=Y$, $f_i(0)=R_i$ and $f_i(1)=Y_i$, and that this solution has value $ \sum_{t=k}^n {n \choose t} p^t$. Indeed, we can see that 

\[
\tr \left( \sum_{\left(i_1, \ldots,i_n \right) \in \Sigma_k^n} f(i_1) \otimes \ldots \otimes f(i_n) \right) = \sum_{t  = k }^n \left| \Sigma^n_t \right| \tr(Y)^t 1^{n-t} =  \sum_{t=k}^n {n \choose t} p^t.
\]

As far as feasibility is concerned, for all the conditions except the last $\geq$ inequality in the Dual Problem 5, it follows that they are satisfied using the same analysis as the one we performed in Chapter \ref{ch:avg} to prove that our solution to Dual Problem 3 was feasible. For the last condition, we have that as

\[
f_{r}(1) \otimes \mathbb{I}_{\Y_r}=R_r  \geq P_1 \text{    and    } f_{r}(0) \otimes \mathbb{I}_{\Y_r} =  P_0 + P_1 \geq P_0,
\]

\noindent then

\[
\sum_{\left( i_1, \ldots, i_n \right) \in \Sigma_{\geq k}^n}  f_r(i_1) \otimes \ldots \otimes f_r(i_n) \geq \sum_{\left( i_1, \ldots, i_n \right) \in \Sigma_{\geq k}^n}  P_{i_1} \otimes \ldots \otimes P_{i_n}.
\]

This is because then the terms at the left hand side of the last constraint can be paired to the terms at the right hand side, in such a way that the term at the left hand side is $\geq$ than the term at the right hand side. 

To improve the analysis, one approach would be then obtaining a solution in which the terms at the left hand side in the last constraints can still be paired with the terms at the right hand side in this way (and the other constraints are also still satisfied), but the trace for the operator in our solution belonging to $\pos{\X_1}$ is smaller. However, attempts at that approach have been unsuccessful so far at giving us a better value as a function of $p$. A possible way of doing so would be letting $Y', \{Y_i'\}$ denote  a solution to the dual problem corresponding to the situation in which Bob is trying to maximize his probability of obtaining outcome $0$, and then letting $f(0)$ be $Y'$, and $f_i(0)$ be $Y_i'$. Then, as 

\[
 f_{r}(0) \otimes \mathbb{I}_{\Y_r} =  Y_r' \otimes  \mathbb{I}_{\Y_r}  \geq P_0,
\]

\noindent  we would still have a solution that satisfies the last constraint (and as can be checked again with an analysis similar to the one in Chapter \ref{ch:avg}, the other constraints are satisfied as well). The value of this solution would be 

\[
 \sum_{t=k}^n {n \choose t} p^t \tr(Y')^{n-t}.
\]

\noindent However, the value of $\tr(Y')$  does not in principle bear any relationship with $p$, and it might as well be $1$, so this is not necessarily a better bound than the one we already obtained.

 We can try then an alternative approach, in which we assign several terms of the right hand side in the last inequality to a term in the left hand side. For example, if we had $n=3$ and $k=2$, in the solution to Dual Problem 5 corresponding to our bound of $\sum_{t=k}^n {n \choose t} p^t$, the right hand side of the last constraint would contain $P_0 \otimes P_1 \otimes P_1$ and  $P_1 \otimes P_1 \otimes P_1$, matched at the left hand side  by $R_r \otimes \mathbb{I}_{\Y_r} \otimes Y_r \otimes \mathbb{I}_{\Y_r}\otimes Y_r \otimes \mathbb{I}_{\Y_r} $ and $Y_r \otimes \mathbb{I}_{\Y_r} \otimes Y_r  \otimes \mathbb{I}_{\Y_r} \otimes Y_r \otimes \mathbb{I}_{\Y_r} $. However, it would be enough to have  $R_r \otimes \mathbb{I}_{\Y_r} \otimes Y_r \otimes \mathbb{I}_{\Y_r}\otimes Y_r \otimes \mathbb{I}_{\Y_r} $ at the left hand side, since  

\[
P_0 \otimes P_1 \otimes P_1 + P_1 \otimes P_1 \otimes P_1 = (P_0+P_1) \otimes P_1 \otimes P_1 = R_r \otimes   \mathbb{I}_{\Y_r}  \otimes P_1 \otimes P_1 
\]

Based on this idea, we build the following solution $S^{n,k}$ to Dual Problem 5, defined recursively as a function of $n$ and $k$:

\begin{mylist}{\parindent}

\item[$\bullet$ ] If $k=0$, then the solution is $S^{n,k}=\rho^{\otimes n},\{ R_i^{\otimes n} \}$

\item[$\bullet$] If $k=n$, then the solution is $S^{n,k}=Y^{\otimes n},\{ Y_i^{\otimes n} \}$

\item[$\bullet$ ] If $0 < k < n$, then the solution is 

\begin{eqnarray*}
S^{n,k} & = &   \left( \rho,\{ R_i\}  \right) \otimes S^{n-1,k}  \\ 
&  & + \left( Y \otimes \sum_{\left( i_1, \ldots, i_{n-1} \right)\in \Sigma_{k-1}^{n-1}} f(i_1) \otimes \ldots \otimes f(i_{n-1}), \left\{ Y_i \otimes \sum_{\left( i_1, \ldots, i_{n-1} \right)\in \Sigma_{k-1}^{n-1}}  f_i(i_1) \otimes \ldots \otimes f_i(i_{n-1}) \right\} \right)\\
\end{eqnarray*}

\noindent where the tensor product with $S^{n-1,k}$ and the sum between the first row and the second row are taken block-wise.

\end{mylist}

\noindent Now, we have that the value of this solution is  $p^k {n \choose k}$. Indeed,

\begin{mylist}{\parindent}

\item[$\bullet$] If $k=0$, then the value is $\tr(\rho)^n=p^0 {n \choose 0}$.

\item[$\bullet$ ] If $k=n$, then the value is $\tr(Y)^n=p^n = p^n {n \choose n} $.

\item[$\bullet$ ] If $0 < k < n$, we can use induction on $n$, with the base case being $n=1$, covered by the previous two cases. We have then that the value is 

\begin{eqnarray*}
& & \tr(\rho) p^k { n -1 \choose k}  + \tr(Y) {n-1 \choose k-1} \tr(Y)^{k-1} \tr(\rho)^{n-k} \\
& = & p^k{ n -1 \choose k} +  p {n-1 \choose k-1} p^{k-1} = p^k { n \choose k}
\end{eqnarray*}

\end{mylist}

We also have that the solution is actually feasible. We prove in the same way as for the previous solutions we consider to Dual Problem 5 that all $\geq$ constraints except the last one are satisfied, using that if one of the blocks of our solution includes  $f_i(i_1) \otimes \ldots \otimes f_i(i_n)$, the previous one will include $f_{i-1}(i_1) \otimes \ldots \otimes f_{i-1}(i_n)$ (with $f$ instead of $f_{i-1}$ if $i = 2$). For the last constraint, we have that 

\begin{mylist}{\parindent}

\item[$\bullet$ ] If $k=0$, then 

\[
R_r^{\otimes n} \otimes \mathbb{I}_{\Y_r^{\otimes n} } = \left( P_0 + P_1 \right)^{\otimes n} =    \sum_{\left( i_1, \ldots, i_n \right) \in \Sigma_{\geq 0}^n} P_{i_1} \otimes \ldots \otimes P_{i_n}
\]

\item[$\bullet$] If $k=n$, then

\[
Y_r^{\otimes n} \otimes  \mathbb{I}_{\Y_r^{\otimes n} }  = (Y_r \otimes   \mathbb{I}_{\Y_r}) ^{\otimes n} \geq P_1 ^ {\otimes n} = \sum_{\left( i_1, \ldots, i_n \right) \in \Sigma_{\geq n}^n} P_{i_1} \otimes \ldots \otimes P_{i_n}
\]

\item[$\bullet$] If $0 < k < n$, we can use induction on $n$, in the same way as in our calculation of the value of the solution. We have then that, using the fact that $S^{n-1,k}$ is feasible for the corresponding program, the value at the left hand side of the last constraint is 

\begin{eqnarray*}
 &  \geq & R_r \otimes  \mathbb{I}_{\Y_r} \otimes  \sum_{\left( i_1, \ldots, i_{n-1} \right)\in \Sigma_{\geq k}^{n-1}}  P_{i_1} \otimes \ldots \otimes P_{i_{n-1}}  \\
 & & + Y_r  \otimes \sum_{\left( i_1, \ldots, i_{n-1} \right)\in \Sigma_{k-1}^{n-1}}  f_i(i_1) \otimes \ldots \otimes f_i(i_{n-1})  \otimes \mathbb{I}_{\Y_r^{\otimes n}}   \\
  & \geq  & (P_0 + P_1) \otimes  \sum_{\left( i_1, \ldots, i_{n-1} \right)\in \Sigma_{\geq k}^{n-1}}  P_{i_1} \otimes \ldots \otimes P_{i_{n-1}} + P_1 \otimes \sum_{\left( i_1, \ldots, i_{n-1} \right)\in \Sigma_{k-1}^{n-1}}  P_{i_1} \otimes \ldots \otimes P_{i_{n-1}}    \\
  & =  & \sum_{\left(i_1, \ldots, i_{n} \right) \in \{0,1\} \times \Sigma_{\geq k}^{n-1} \cup \{1\} \times  \Sigma_{k-1}^{n-1}  }  P_{i_1} \otimes \ldots \otimes P_{i_{n}}   \\
\end{eqnarray*}

Decomposing $\Sigma_{\geq k}^n$ into two subsets, the one with at least $k$ $1$s in the last $n-1$ symbols, and the one with exactly $k-1$ $1$s in the last $n-1$ symbols, we have that the last formula in our chain of inequalities is indeed

\[
 \sum_{\left( i_1, \ldots, i_n \right) \in \Sigma_{\geq k}^{n}}  P_{i_1} \otimes \ldots \otimes P_{i_n},
\]

\noindent so our solution to Dual Problem 5 satisfies the last constraint, as desired.
 
\end{mylist}

%% file: 07-errorreduction.tex
\newpage

\chapter{Error reduction for interactive proof systems}\label{ch:error}

An interactive proof system is a situation in which an object belonging to one of two disjoint sets $ \left( L_{\text{yes}},L_{\text{no}} \right)$ is known to two parties, one of which is trying to convince the other that the object belongs in $L_{yes}$. If the object does indeed belong in $L_{yes}$, the probability that the second individual is successfully convinced will be higher than if it does not. As it is standard to do in theoretical computer science, we will assume that the objects in the sets $(L_{\text{yes}}, L_{\text{no}})$ are modelled as binary strings. Several variations of this setting have been widely studied in complexity theory (see e.g. \cite{BabaiM88,GoldwasserMR89} for two foundational papers in the area), as it is possible to defined complexity classes in terms of sets of objects for which such an interaction exists.  

The model for these interactions is similar to the one we have considered in this thesis. In this chapter, we will then go back to referring to the party we have called Alice as the \text{verifier} and the party we have called Bob as the \emph{prover}. There will be two outcomes for the interaction between them, one of them called the \emph{accepting} outcome, and the other one called the \emph{rejecting} outcome. We say that the verifier accepts whenever the outcome of the interaction is the accepting outcome, and that the verifier rejects whenever the outcome of the interaction is the rejecting outcome. We also place the additional restriction that the process by which the verifier operates must be an efficient process, so its computational ability is restricted to  quantum  (or probabilistic, in the classical case) polynomial time in the size of the shared object. The prover's computational ability is still unrestricted.

For an interactive proof system to be good, it should be possible for the verifier to make a reasonable guess about whether $x \in L_{\text{yes}}$ from the outcome of the interaction. Then, we say that an interactive proof is valid for a problem specified by $(L_{\text{yes}}, L_{\text{no}})$, with parameters $\alpha$ and $\beta$, $\beta < \alpha$, whenever 

\begin{mylist}{\parindent}
\item[1.]
  If $x \in L_{\text{yes}}$, it is possible for the prover to
  convince the verifier to accept with probability at least $\alpha$. This is called the \emph{completeness} condition, and corresponds to the condition  in formal logic that true statements can be proved. $1-\alpha$ is called then the \emph{completeness error}.

\item[2.]
    If $x \in L_{\text{no}}$, the verifier always accepts with
  probability at most $\beta$, regardless of the prover's actions. This is called the \emph{soundness} condition, and corresponds to the condition in formal logic that false statements cannot be proved. $\beta$ is called then the \emph{soundness error}.

\end{mylist}

We might have for example that $\alpha = 1/2+\delta$ and $\beta=1/2-\delta$, for some small $\delta > 0$. However, the verifier would be able to make a  better guess about whether $x \in L_{\text{yes}}$ from the outcome of the interaction if we had $\alpha = 1 - \epsilon$ and $\beta = \epsilon$, for a small value of $\epsilon>0$. The process of specifying a new interactive proof system from another one in a way that improves on the value of $\alpha$ and $\beta$ is called then  \emph{error reduction}.

A natural procedure to perform for error reduction would be the same one that is usually performed in the case of probabilistic algorithms. That is, the verifier could repeat the interaction several times, and accept if and only if the number of  accepting outcomes that are obtained is above a certain threshold. In the situation under consideration, one is to understand that it is important for the new verifier to run these independent tests in parallel (as opposed to requiring the prover to respond sequentially to the individual tests). A motivation for this comes for the fact that in the complexity classes defined in terms of quantum interactive proofs, the number of rounds is often considered to be a fixed constant, so one increasing the number of rounds might not be a possibility.

However, the analysis that proves that this intuitive procedure works in the case of probabilistic algorithms relies in the fact that the analysis of the different repetitions can be made in an independent way. It is not clear that we could do this in our analysis, since as we saw in Chapter \ref{ch:nindep},  it might not be optimal for a hypothetical prover that interacts with many independent executions of an interactive proof system to respect the independence of these executions when the objective of the prover is to get a number of accepting outcomes past a certain threshold.

 Note however that, as we saw in Chapter \ref{ch:nindep}, in the classical case we can indeed assume that the prover respects the independence of the executions. And it is indeed well-known in that case that the same argument that is used for probabilistic algorithms can be extended, and error reduction based on parallel repetition and a threshold value computation works perfectly
for (single-prover) interactive proof systems.  \footnote{%
  The situation is very different for \emph{multi-prover} interactive
  proof systems, wherein the subject of parallel repetition is
  complicated \cite{Raz98,Holenstein09,Raz08}.}  By this we mean that not only parallel repetition and a threshold value computation can be used for error reduction,  but that as it follows from the behaviour in the classical case that we described in Chapter \ref{ch:nindep}, we have the stronger statement that if $p$ is the optimal probability for the original verifier to obtain an accepting outcome for some $x$, then the optimal probability to cause at least $t$ acceptances out of $k$ independent repetitions of the original interaction is
\[
\sum_{j = t}^k \binom{k}{j} p^j (1 - p)^{k-j}.
\]

\noindent Using this and standard Chernoff bounds, we have that our suggested strategy for error reduction does quickly reduce the error. Our example in Chapter \ref{ch:nindep} shows that this perfect behaviour for parallel repetition does not always hold in the quantum case. However, it might still be the case that an strategy based on parallel repetition and a threshold value computation can be used for error reduction. This would provide a simpler strategy for performing error reduction in quantum interactive proofs that the ones that are known in the literature \cite{JainUW09,KitaevW00}. 

We will show now then, using our results from Chapter \ref{ch:ub}, that the natural procedure that we suggest for error reduction does indeed work for a certain range of values for the $\alpha$ and $\beta$ parameters. More formally, we prove the following Theorem:

\begin{theorem}  \label{th:er}
Let the parameters $\alpha$ and $\beta$ for a quantum interactive proof system be constant real numbers, with $0 \leq \beta < 2^{-\frac{H(\alpha)}{\alpha}} < \alpha \leq 1$. Then, a strategy based on parallel repetition followed by a threshold value computation will bring the soundness and completeness errors below $\epsilon$ in $O(\log \frac{1}{\epsilon})$ rounds.
\end{theorem}

\begin{proof}
Let $p$ be the optimal probability for the prover to obtain an accepting outcome with the case, and $c$ be a constant rational number $\frac{c1}{c2}$ strictly smaller than $\alpha$ (we will further restrict the value of $c$ later). We will let the threshold for the error reduction procedure be $k=\floor{c n}$. 

We start by looking at the completeness error, corresponding to the situation in which $x \in L_{\text{yes}}$ and $p \geq \alpha$. Consider an strategy for the prover which just plays the optimal strategy for a single repetition independently in each of the independent interactions. The probability that this strategy obtains a given number of accepting outcomes will be given then by a binomial distribution with parameters $p$ and $n$, and this distribution follows the Chernoff bound

\[
P(X \leq pn(1-\lambda) ) \leq \exp \left( -\frac{pn \lambda^2 }{2} \right)
\]

As the probability that the number of accepting outcomes falls below the threshold is equal to the probability that it is at most $cn$, we have then that this probability is bounded by

\[
\exp \left( \frac{-pn \left( 1 - \frac{cn}{pn} \right)^2}{2} \right) = \exp \left( \frac{-pn \left( 1 - \frac{c}{p} \right)^2}{2} \right)
\]

As $c < \alpha \leq p $, this is an exponentially decreasing function of $n$, so it is indeed enough to repeat the interaction in parallel $O(\log \frac{1}{\epsilon})$ times to obtain a completeness error below $\epsilon$.

We look now at the soundness error,  corresponding to the situation in which $x \in L_{\text{no}}$ and $ p \leq \beta$.  Then, from our results in Chapter \ref{ch:ub}, we have that the probability that the prover can obtain a number of accepting outcomes above the threshold is upper bounded by

\[
p^{k} {n \choose k }
\]

To analyze this expression, we take its logarithm, which is equal to

\[
 k \lg p + \lg n! - \lg k! - \lg \left( n- k \right)!
\]

Now, we can obtain lower and upper bounds for $\lg n! = \sum_{i=1}^n \lg i$ by integrating $\lg$. The lower bound is $n \lg n - \frac{n+1}{\ln 2}$, while the upper bound is $(n+1) \lg (n+1) - \frac{n}{\ln 2}$.  Using these bounds, we obtain an upper bound on the previous expression of

\begin{eqnarray*}
& &  p \lg k  + (n+1) \lg( n+1) -\frac{n}{\ln 2} - k \lg k +\frac{k+1}{\ln 2} -(n-k)\lg(n-k) + \frac{n-k+1}{\ln 2}   \\
& = &   p \lg k  + \lg (n+1) + n\lg( n+1) -  k  \lg  k  -(n- k )\lg(n- k ) + \frac{2}{\lg 2}
\end{eqnarray*}

If we write $k=\floor{cn} = \floor{ \frac{c_1}{c_2} n }$ as $ cn - \frac{c_1 n \mod c_2}{c_2} = n \left(c - \frac{c_1 n \mod c_2}{c_2 n}\right)$, this is equal to 

\begin{eqnarray*}
& &   \left( cn - \frac{c_1 n \mod c_2}{c_2} \right) \lg p + \lg(n+1) + n\lg n + n \lg(1 + \frac{1}{n}) - \left( cn - \frac{c_1 n \mod c_2}{c_2} \right)   \lg n  \\\
& & -  \left( cn - \frac{c_1 n \mod c_2}{c_2} \right)  \lg \left( c - \frac{c_1 n \mod c_2}{c_2n} \right) -  \left( (1-c)n + \frac{c_1 n \mod c_2}{c_2} \right)   \lg n   \\
& & -     \left( (1-c)n + \frac{c_1 n \mod c_2}{c_2} \right)   \lg  \left( (1-c) + \frac{c_1 n \mod c_2}{c_2n} \right) + \frac{2}{\lg 2} \\
\end{eqnarray*}

We can see that the terms in $n \log n$ cancel each other, and the previous expression can then be written as 

\begin{eqnarray*}
 & &n \left[ c \lg p + \lg(1 + \frac{1}{n})  -c  \lg \left( c - \frac{c_1 n \mod c_2}{c_2n} \right)  - (1-c) \lg  \left( (1-c) + \frac{c_1 n \mod c_2}{c_2n}  \right) \right] \\
& &-   \frac{c_1 n \mod c_2}{c_2}  \lg p + \lg(n+1) + \frac{c_1 n \mod c_2}{c_2}   \lg \left( c - \frac{c_1 n \mod c_2}{c_2n} \right) \\
 & &-  \frac{c_1 n \mod c_2}{c_2}   \lg  \left( (1-c) + \frac{c_1 n \mod c_2}{c_2n} \right) \\
\end{eqnarray*}

\noindent that is, as 

\begin{eqnarray*}
n \left[ c \lg p+  \lg(1 + \frac{1}{n})  -c  \lg \left( c - \frac{c_1 n \mod c_2}{c_2n} \right)  - (1-c) \lg  \left( (1-c) + \frac{c_1 n \mod c_2}{c_2n}  \right) \right]  + o(n)
\end{eqnarray*}

As $n$ goes to infinity, the coefficient for $n$ goes to $c \lg p -H(c)$, which will be negative if $\frac{H(c)}{c} > \lg p$.  Now, we have $\lg p \leq \lg \beta < \frac{H(\alpha)}{\alpha} $. As $\frac{H(x)}{x}$ is a continuous function in the interval $(0,1]$, we can then pick $c$ as a constant close enough to $\alpha$ that $\lg \beta < \frac{H(c)}{c}$, and therefore $\frac{H(c)}{c} > \lg p$.  We have then that there is a positive integer constant $n_1$ and a positive real constant $\lambda_1$  such that for all $n \geq n_1$ , the coefficient for $n$ in the previous expression is upper-bounded by $\lambda_1$. Taking also into account the $o(n)$ term, we have then that for $n$ greater or equal than a constant $n_2$, the logarithm of our bound of the soundness error  is upper bounded by $-\lambda_2 n$ for some positive real $\lambda_2$, obtaining then the asymptotic result that the soundness error can be reduced to $\epsilon$ in $O(\log \frac{1}{\epsilon})$ rounds.

\end{proof}

\begin{figure}[t]
  \begin{center}
 	\includegraphics[height=5cm]{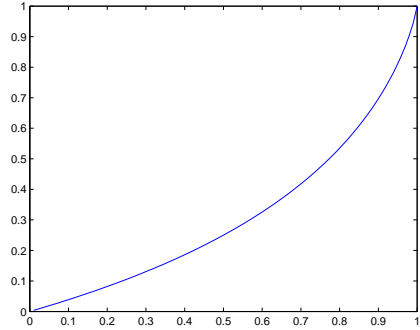}
  \end{center}
  \caption{Graph for $2^\frac{-H(x)}{x}$}
  \label{fig:entropy}
\end{figure}

We can see in Figure \ref{fig:entropy} that the condition $\beta < 2^{-\frac{H(\alpha)}{\alpha}}$ is meaningful, in the sense that there seems to be a wide range of values of $\beta$ and $\alpha$ for which it holds that $\beta < 2^{-\frac{H(\alpha)}{\alpha}}$. Indeed,  it is possible to prove  that $ 2^{-\frac{H(\alpha)}{\alpha}} > \alpha/3$, so for $\beta < \alpha/3$ the condition $\beta < 2^{-\frac{H(\alpha)}{\alpha}}$  will hold.

%% file: 08-conclusion.tex
\newpage

\chapter{Conclusion}\label{ch:concl}

This thesis has considered several questions related to the parallel repetition of a simple kind of interaction, broadly centered around the relevance of correlations arising in the quantum information theoretic versions of these interactions.

We have seen in Chapter \ref{ch:avg} how the presence of those correlations does not affect the optimality for Bob of acting independently in the different repetitions whenever he associates a value to each outcome, and is trying to optimize the value obtained by repetition. However, when Bob is trying to force a certain outcome to occur a number of times past a given threshold, then we have seen in Chapter \ref{ch:nindep} that the correlations that are possible between the actions of Bob for different repetitions can give rise to a strikingly non-classical hedging type of behaviour.

Our work may have then relevance in settings considered in cryptography,
where certain protocols might very well be abstracted as tests of the
sort we have considered (this is the case, for example, for quantum money  \cite{MolinaVW12}, and quantum coin-flipping \cite{GutoskiW07}). The extent to which a dishonest individual can attack such protocols
by correlating independent executions is an important security
consideration that some would-be cryptographers might fail to
consider. Our results in Chapter \ref{ch:nindep} demonstrate then that quantum attacks to such protocols may
exhibit striking non-classical and counter-intuitive properties, and
should therefore be given very careful consideration.

We have also  established in Chapter \ref{ch:nindep}  certain conditions, including the classical case, under which it is optimal for Bob to act independently in different repetitions. Then, we have established in Chapter \ref{ch:ub} certain general quantitative bounds to the hedging type of behaviour that we have observed. And finally, we have discussed in Chapter \ref{ch:error} the connection of our results with certain techniques for error reduction in quantum interactive proofs.

Some questions left open in our work and that remain to be answered are the following ones:

\begin{mylist}{\parindent}

\item[1. ] 
Is it possible to use separability properties of the operators in our semidefinite programs to obtain bounds on hedging behaviours?

\item[2. ]  
Is it possible to improve our bounds in Chapter \ref{ch:ub} concerning  the optimum probability for Bob of achieving a winning outcome in at least $k$ of the  $n$ interactions? 

\item[3. ]  
If the answer to the previous question is positive, can these improvements be used to prove that the naive way of parallel repetition discussed in Chapter \ref{ch:error} does always work?

\item[4. ]  
If the answer to the previous questions is positive, how fast does the naive way of parallel repetition reduce the error?

\end{mylist}

Ideally, there would be an exponential decay in the optimum probability for Bob of achieving a winning outcome in a fraction of the interactions above $p n$, as a function of the number of repetitions $n$. That would answer both of the first questions in a positive way, and would establish that the naive way of parallel repetition does indeed reduce the error quickly in an asymptotical sense.